\documentclass[12pt,onecolumn,draftclsnofoot]{IEEEtran}
\IEEEoverridecommandlockouts

\usepackage{times}
\usepackage[final]{graphicx}
\usepackage{amsmath,amsfonts}
\usepackage{amsthm}
\usepackage{amssymb,amsbsy}

\usepackage{cite}
\usepackage{multirow}

\newcommand{\ignore}[1]{}

\newcommand{\hspp}{\hspace{0.05in} }
\newcommand{\hsppp}{\hspace{0.02in} }

\newtheorem{prop}{Proposition}
\newtheorem{thm}{Theorem}

\newsavebox{\savepar}

\begin{document}
\title{Beamforming Tradeoffs for Initial UE Discovery in Millimeter-Wave MIMO Systems}
\author{\large Vasanthan Raghavan, Juergen Cezanne, Sundar Subramanian, Ashwin Sampath,
and Ozge Koymen \\
Qualcomm Corporate R\&D, Bridgewater, NJ 08807 \\
E-mail: \{vraghava, jcezanne, sundars, asampath, okoymen@qti.qualcomm.com\}
}

\maketitle
\vspace{-10mm}

\begin{abstract}
\noindent
Millimeter-wave (mmW) multi-input multi-output (MIMO) systems have gained
increasing traction towards the goal of meeting the high data-rate requirements
for next-generation wireless systems. The focus of this work is on low-complexity
beamforming approaches for initial user equipment (UE) discovery in such systems.
Towards this goal, we first note the structure of the optimal beamformer with
per-antenna gain and phase control and establish the structure of good beamformers
with per-antenna phase-only control. Learning these right singular vector (RSV)-type
beamforming structures in mmW systems is fraught with considerable complexities such
as the need for a non-broadcast system design, the sensitivity of the beamformer
approximants to small path length changes, inefficiencies due to power amplifier
backoff, etc. To overcome these issues, we establish a physical interpretation
between the RSV-type beamformer structures and the angles of departure/arrival (AoD/AoA)
of the dominant path(s) capturing the scattering environment. This physical
interpretation provides a theoretical underpinning to the emerging interest on
directional beamforming approaches that are less sensitive to small path length changes.
While classical approaches for direction learning such as MUltiple SIgnal Classification
(MUSIC) have been well-understood, they suffer from many practical difficulties in a mmW
context such as a non-broadcast system design and high computational complexity. A simpler
broadcast-based solution for mmW systems is the adaptation of limited feedback-type
directional codebooks for beamforming at the two ends. We establish fundamental limits
for the best beam broadening codebooks and propose a construction motivated by a virtual
subarray architecture that is within a couple of dB of the best tradeoff curve at all
useful beam broadening factors. We finally provide the received ${\sf SNR}$ loss-UE
discovery latency tradeoff with the proposed beam broadening constructions. Our results
show that users with a reasonable link margin can be quickly discovered by the proposed
design with a smooth roll-off in performance as the link margin deteriorates. While these
designs are poorer in performance than the RSV learning approaches or MUSIC for cell-edge
users, their
low-complexity that leads to a broadcast system design makes them a useful candidate for
practical mmW systems.
\end{abstract}


\section{Introduction}
\label{sec1}
The ubiquitous nature of communications made possible by the smart-phone and social
media revolutions has meant that the data-rate requirements will continue to grow at an
exponential pace. On the other hand, even under the most optimistic assumptions, system
resources can continue to scale at best at a linear rate, leading to enormous mismatches
between supply and demand. Given this backdrop, many candidate solutions have been
proposed~\cite{rusek,qualcomm,boccardi1} to mesh into the patchwork that addresses the
$1000$-{\sf X} data challenge~\cite{qualcomm1} --- an intermediate stepping stone towards
bridging this burgeoning gap.

One such solution that has gained increasing traction over the last few years is
communications over the millimeter-wave (mmW)
regime~\cite{khan,torkildson,rappaport,rangan,roh,akdeniz} where the carrier frequency
is in the $30$ to $300$ GHz range. Spectrum crunch, which is the major bottleneck at
lower/cellular carrier frequencies, is less problematic at higher carrier frequencies
due to the availability of large (either unlicensed or lightly licensed) bandwidths.
However, the high frequency-dependent propagation and shadowing losses (that can offset
the link margin substantially) complicate the exploitation of these large bandwidths. It
is visualized that these losses can be mitigated by limiting coverage to small areas and
leveraging the small wavelengths that allows the deployment of a large number of antennas
in a fixed array aperture.

Despite the possibility of multi-input multi-output (MIMO) communications, mmW signaling
differs significantly from traditional MIMO architectures at cellular frequencies. Current
implementations\footnote{In a downlink setting, the first dimension corresponds to the
number of antennas at the user equipment end and the second at the base-station end.} at
cellular frequencies are on the order of $4 \times 8$ with a precoder rank (number of
layers) of $1$ to $4$; see, e.g.,~\cite{lim}. Higher rank signaling requires multiple
radio-frequency (RF) chains\footnote{An RF chain includes (but is not limited to)
analog-to-digital and digital-to-analog converters, power and low-noise amplifiers,
upconverters/mixers, etc.} which are easier to realize at lower carrier frequencies
than at the mmW regime. Thus, there has been a growing interest in understanding the
capabilities of low-complexity approaches such as beamforming (that require only a
single RF chain) in mmW
systems~\cite{torkildson,venkateswaran,brady,oelayach,hur,alkhateeb,sun}. On the other
hand, smaller form factors at mmW frequencies ensure\footnote{For example, an $N_t = 64$
element uniform linear array (ULA) at $30$ GHz requires an aperture of $\sim1$ foot at
the critical $\lambda/2$ spacing --- a constraint that can be realized at the base-station
end.} that configurations such as $4 \times 64$ (or even higher dimensionalities) are
realistic. Such high antenna dimensionalities as well as the considerably large bandwidths
at mmW frequencies result in a higher resolvability of the multipath and thus, the MIMO
channel is naturally {\em sparser} in the mmW regime than at cellular
frequencies~\cite{akdeniz,zhao,oelayach,akbar,vasanth_jstsp,vasanth_it2}.

While the optimal right singular vector (RSV) beamforming structure has been known in
the MIMO literature~\cite{tky_lo}, an explicit characterization of the connection of
this structure to the underlying physical scattering environment has not been
well-understood. We start with such an explicit physical interpretation by showing
that the optimal beamformer structure corresponds to beam steering across the different
paths with appropriate power allocation and phase compensation confirming many recent
observations~\cite{oelayach,oelayach1}. Despite using only a single RF chain, the
optimal beamformer requires per-antenna phase and gain control (in general), which
could render this scheme disadvantageous from a cost perspective. Thus, we also
characterize the structure and performance of good beamformer structures with
per-antenna phase-only control~\cite{hur,david_egc,pi,oelayach1}.

Either of these structures can be realized in practice via an (iterative) RSV learning
scheme. To the best of our understanding, specific instantiations of RSV learning such
as power iteration 
have not been
studied in the literature (even numerically), except in the noise-less case~\cite{dahl}.
A low-complexity proxy to the RSV-type beamformer structures is directional beamforming
along the dominant path at the millimeter-wave base-station (MWB) and the user equipment
(UE). The directional beamforming structure is particularly relevant in mmW systems due
to the sparse nature of the channel, and this structure is not expected to be optimal at
non-sparse cellular narrowband frequencies. Our studies show that directional beamforming
suffers only a minimal loss in performance relative to the optimal structures, rendering
the importance of direction learning for practical mmW MIMO systems, again confirming many
recent observations~\cite{oelayach,oelayach1,rangan,swindlehurst_mmw,brady,hur,alkhateeb,sun}.

Such schemes can be realized in practice via direction learning techniques. Direction
learning methods such as MUltiple SIgnal Classification (MUSIC), Estimation of Signal
Parameters via Rotational Invariance Techniques (ESPRIT) and related
approaches~\cite{schmidt,roy_kailath,krim} have been well-understood in the signal
processing literature, albeit primarily in the context of military/radar applications.
Their utility in these applications is as a ``super-resolution'' method to discern
multiple obstacles/targets given that the pre-beamforming ${\sf SNR}$ is moderate-to-high,
at the expense of array aperture (a large number of antennas) and computations/energy.
While MUSIC has been suggested as a possible candidate beamforming strategy for mmW
applications, it is of interest to fairly compare different beamforming strategies given
a specific objective (such as initial UE discovery).

The novelty of this work is on such a fair comparison between different strategies in
terms of: i) architecture of system design (broadcast/unicast solution), ii)
resilience/robustness to low pre-beamforming ${\sf SNR}$'s expected in mmW systems,
iii) performance loss relative to the optimal beamforming scheme, iv) adaptability to
operating on different points of the tradeoff curve of initial UE discovery latency
vs.\ accrued beamforming gain, v) scalability to beam refinement as a part of the data
transfer process, etc. Our broad conclusions are as follows. The fundamental difficulty
with any RSV learning scheme is its extreme sensitivity to small path length changes
that could result in a full-cycle phase change across paths, which becomes increasingly
likely at mmW carrier frequencies. Further, these methods suffer from implementation
difficulties as seen from a system-level standpoint such as a non-broadcast design,
poor performance at low link margins, power amplifier (PA) backoff, poor adaptability to
different beamforming architectures, etc. In commonality with noisy power iteration,
MUSIC, ESPRIT and related approaches also suffer from a non-broadcast solution, poor
performance at low link margins and computational complexity.

To overcome these difficulties, inspired by the limited feedback
literature~\cite{david_review,david_grass,mukkavilli}, we study the received
${\sf SNR}$ performance with the use of a globally known directional beamforming
codebook at the MWB and UE ends. The simplest codebook of beamformers made of constant
phase offset (CPO) array steering vectors (see Fig.~\ref{fig4c} for illustration)
requires an increasing number of codebook elements as the number of antennas increases
to cover a certain coverage area, thereby corresponding to a proportional increase in
the UE discovery latency. We study the beam broadening problem of trading off UE
discovery latency at the cost of the peak gain in a beam's coverage
area~\cite{hur,rajagopal}. We establish fundamental performance limits for this problem,
as well as realizable constructions that are within a couple of dB of this limit at all
beam broadening factors (and considerably lower at most beam broadening factors). With
beams so constructed, we show that directional beamforming can tradeoff the UE discovery
latency substantially for a good fraction of the users with a slow roll-off in performance
as the link margin deteriorates. While the codebook-based approach is sub-optimal relative
to noisy power iteration or MUSIC, its simplicity of system design and adaptability to
different beamforming architectures and scalability to beam refinement makes it a viable
candidate for initial UE discovery in practical mmW beamforming implementations. Our work
provides further impetus to the initial UE discovery problem (See Sec.~\ref{sec_5e} for a
discussion) that has attracted attention from many related recent
works~\cite{jeong_et_al,ghadikolaei,barati}.

\noindent {\bf \em \underline{Notations:}} Lower- (${\bf x}$) and upper-case block
(${\bf X}$) letters denote vectors and matrices with ${\bf x}(i)$ and ${\bf X}(i, j)$
denoting the $i$-th and $(i, j)$-th entries of ${\bf x}$ and ${\bf X}$, respectively.
$\| {\bf x} \|_2$ and $\| {\bf x} \|_{\infty}$ denote the $2$-norm and $\infty$-norm
of a vector ${\bf x}$, whereas ${\bf x}^H$, ${\bf x}^T$ and ${\bf x}^{\star}$ denote
the complex conjugate Hermitian transpose, regular transpose and complex
conjugation operations of ${\bf x}$, respectively. We use ${\mathbb{C}}$ to denote the
field of complex numbers, ${\mathbb{E}}$ to denote the expectation operation and
$\chi({\cal A})$ to denote the indicator function of a set ${\cal A}$.

\section{System Setup}
\label{sec2}

We consider the downlink setting where the MWB is equipped with $N_t$ transmit antennas
and the UE is equipped with $N_r$ receive antennas. Let ${\bf H}$ denote the $N_r \times N_t$
channel matrix capturing the scattering between the MWB and the UE. We are interested in
beamforming (rank-$1$ signaling) over ${\bf H}$ with the unit-norm $N_t \times 1$
beamforming vector ${\bf f}$. The system model in this setting is given as
\begin{eqnarray}
{\bf y} = \sqrt{ \rho_{\sf f} } \cdot {\bf H} {\bf f} s + {\bf n}
\label{eq1}
\end{eqnarray}
where $\rho_{\sf f}$ is the pre-beamforming\footnote{The pre-beamforming ${\sf SNR}$
is the received ${\sf SNR}$ seen with antenna selection at both ends of the link and
under the wide-sense stationary uncorrelated scattering (WSSUS) assumption. This
${\sf SNR}$ is the same independent of which antenna is selected at either end.}
${\sf SNR}$, $s$ is the symbol chosen from an appropriate
constellation for signaling with ${\bf E}[s] = 0$ and ${\bf E}[|s^2|] = 1$, and
${\bf n}$ is the $N_r \times 1$ proper complex white
Gaussian noise vector (that is, ${\bf n} \sim {\cal CN}( {\bf 0}, \hsppp {\bf I})$)
added at the UE. The
symbol $s$ is decoded by beamforming at the receiver along the $N_r \times 1$
unit-norm vector ${\bf g}$ to obtain
\begin{eqnarray}
\widehat{s} = {\bf g}^H {\bf y} = \sqrt{ \rho_{\sf f} } \cdot {\bf g}^H {\bf H} {\bf f} s
+ {\bf g}^H {\bf n}.
\label{eq2}
\end{eqnarray}

For the channel, we assume an extended Saleh-Valenzuela geometric model~\cite{saleh}
in the ideal setting where ${\bf H}$ is determined by scattering over
$L$ clusters/paths with no near-field impairments at the UE end and is denoted as follows:
\begin{eqnarray}
{\bf H} = \sqrt{ \frac{N_r N_t}{L} } \cdot \sum_{\ell = 1}^L \alpha_{\ell} \cdot
{\bf u}_{\ell} {\bf v}_{\ell}^H.
\label{eq_1}
\end{eqnarray}
In~(\ref{eq_1}), $\alpha_{\ell} \sim {\cal CN}(0, 1)$ denotes the complex
gain\footnote{We assume a complex Gaussian model for $\alpha_{\ell}$ only for the
sake of illustration of the main results. However, all the results straightforwardly
carry over to more general models.}, ${\bf u}_{\ell}$ denotes the $N_r \times 1$
receive array steering vector, and ${\bf v}_{\ell}$ denotes the $N_t \times 1$
transmit array steering vector, all corresponding to the $\ell$-th path. With this
assumption, the normalization constant $\sqrt{ \frac{N_r N_t}{L} }$ in ${\bf H}$
ensures that the standard channel power normalization\footnote{The standard
normalization that has been used in MIMO system studies is ${\bf E} \left[
{\sf Tr}( {\bf H} {\bf H}^H ) \right] = N_rN_t$. However, as $\{ N_r, N_t \}$
increases as is the case with massive MIMO systems such as those in mmW signaling,
this normalization violates physical laws and needs to be modified appropriately;
see~\cite{vasanth_it2,sayeed_itw} and references therein. Such a modification will
not alter the results herein since the main focus is on a performance comparison
between different schemes. Thus, we will not concern ourselves with these technical
details here.} in MIMO system studies holds. As a typical example of the case where
a uniform linear array (ULA) of antennas are deployed at both ends of the link (and
without loss of generality pointing along the ${\sf X}$-axis in a certain global reference
frame), the array steering vectors ${\bf u}_{\ell}$ and ${\bf v}_{\ell}$ corresponding
to angle of arrival (AoA) $\phi_{{\sf R}, \ell}$ and angle of departure (AoD)
$\phi_{ {\sf T}, \ell}$ in the azimuth in that reference frame (assuming an elevation
angle $\theta_{ {\sf R}, \ell} = \theta_{ {\sf T}, \ell} = 90^{\sf o}$) are given as
\begin{eqnarray}
{\bf u}_{\ell} & = & \frac{1}{ \sqrt{N_r} } \cdot \Big[1 , \hspp
e^{j k d_{\sf R} \cos(\phi_{ {\sf R}, \ell })}, \hspp \cdots \hspp ,
e^{j (N_r - 1) k d_{\sf R} \cos(\phi_{ {\sf R}, \ell })} \Big]^T
\label{eq4}
\\
{\bf v}_{\ell} & = & \frac{1}{ \sqrt{N_t} } \cdot \Big[1 , \hspp
e^{j k d_{\sf T} \cos(\phi_{ {\sf T}, \ell })}, \hspp \cdots \hspp ,
e^{j (N_t - 1) k d_{\sf T} \cos(\phi_{ {\sf T}, \ell })} \Big]^T
\label{eq5}
\end{eqnarray}
where $k = \frac{2\pi}{\lambda}$ is the wave number with $\lambda$ the wavelength of
propagation, and $d_{\sf R}$ and $d_{\sf T}$ are the inter-antenna element
spacing\footnote{With the typical $d_{\sf R} = d_{\sf T} = \frac{\lambda}{2}$
spacing, we have $kd_{\sf R} = kd_{\sf T} = \pi$.} at the receive and transmit sides,
respectively. To simplify the notations and to capture the {\em constant phase offset}
(CPO)-nature of the array-steering vectors and the correspondence with their respective
physical angles, we will henceforth denote ${\bf u}_{\ell}$ and ${\bf v}_{\ell}$
above as ${\sf CPO}(\phi_{ {\sf R},\ell})$ and ${\sf CPO}(\phi_{ {\sf T},\ell})$,
respectively.

\section{Optimal Beamforming and RSV Learning}
\label{sec3}

In terms of performance metric, we are interested in the received ${\sf SNR}$ in
the instantaneous channel setting (that is, ${\bf H} = {\sf H}$), denoted as
${\sf SNR}_{\sf rx}$ and defined as,
\begin{eqnarray}
{\sf SNR}_{\sf rx} \triangleq \rho_{\sf f} \cdot
\frac{ |{\bf g}^H \hsppp {\sf H} \hsppp {\bf f}|^2 \cdot
{\bf E} [ |s|^2 ] }{ {\bf E} \left[ |{\bf g}^H {\bf n}|^2 \right]}
= \rho_{\sf f} \cdot \frac{ |{\bf g}^H \hsppp {\sf H} \hsppp {\bf f}|^2}
{ {\bf g}^H {\bf g} }
\label{eq6}
\end{eqnarray}
since the achievable rate as well as the error probability in estimating $s$ are
captured by this quantity~\cite{david_review,gesbert_review,andrea_review}. We are
interested in studying the performance loss between the optimal beamforming scheme
based on the RSV of the channel and a low-complexity directional beamforming scheme.
Towards this goal, we start by studying the structure of the optimal beamforming
scheme under various RF hardware constraints.
For the link margin, we are interested\footnote{Consider the following
back-of-the-envelope calculation. Let us assume a nominal
path loss corresponding to a $100$-$200$ m cell-radius of $130$ dB, and mmW-specific
shadowing and other losses of $20$ dB. We assume a bandwidth of $500$ MHz with a
noise figure of $7$ dB to result in a thermal noise floor of $-80$ dBm. With an
equivalent isotropically radiated power (EIRP) of $40$ to $55$ dBm in an $N_t = 64$
antenna setting, the pre-beamforming ${\sf SNR}$ corresponding to $64$-level
time-repetition (processing) gain (of $18$ dB) is $-30$ to $-15$ dB. This
suggests that low pre-beamforming ${\sf SNR}$s are the norm in mmW systems.} in low
pre-beamforming ${\sf SNR}$'s.

\subsection{Optimal Beamforming with Full Amplitude and Phase Control}
\label{sec3_a}
We start with the setting where there is full amplitude and phase control of the
beamforming vector coefficients at both the MWB and the UE. Let ${\cal F}_{2}^{N_t}$
denote the class of energy-constrained beamforming vectors reflecting this assumption.
That is, ${\cal F}_{2}^{N_t} \triangleq \left\{ {\bf f} \in {\mathbb{C}}^{N_t} \hsppp
: \hsppp \| {\bf f} \|_2 \leq 1 \right\}$. Under perfect channel state information
(CSI) of ${\sf H}$ at both the MWB and the UE, optimal beamforming vectors
${\bf f}_{\sf opt}$ and ${\bf g}_{\sf opt}$ are to be designed from ${\cal F}_2^{N_t}$ and
${\cal F}_2^{N_r}$ to maximize ${\sf SNR}_{\sf rx}$~\cite{tky_lo}. Clearly,
${\sf SNR}_{\sf rx}$ is maximized with $\| {\bf f}_{\sf opt} \|_2 = 1$, otherwise
energy is unused in beamforming. Further, a simple application~\cite{tky_lo} of the 
Cauchy-Schwarz inequality shows that ${\bf g}_{\sf opt}$ is a matched filter combiner 
at the receiver and assuming that $\| {\bf g}_{\sf opt}\|_2 = 1$ (for convenience), we 
obtain ${\sf SNR}_{\sf rx} = \rho_{\sf f} \cdot {\bf f}_{\sf opt}^H {\sf H}^H {\sf H} 
{\bf f}_{\sf opt}$. We thus have
\begin{eqnarray}
{\bf f}_{\sf opt} = {\sf v}_1, \hspp {\bf g}_{\sf opt} = \frac{ {\sf H} \hsppp {\sf v}_1 }
{ \| {\sf H} \hsppp {\sf v}_1 \|_2},
\label{eq_2}
\end{eqnarray}
where ${\sf v}_1$ denotes a dominant unit-norm right singular vector (RSV) of ${\sf H}$.
Here, the singular value decomposition of ${\sf H}$ is given as ${\sf H} = {\sf U}
\hsppp {\sf \Lambda} \hsppp {\sf V}^H$ with ${\sf U}$ and ${\sf V}$ being $N_r \times N_r$
and $N_t \times N_t$ unitary matrices of left and right singular vectors, respectively,
and arranged so that the corresponding leading diagonal entries of the $N_r \times N_t$
singular value matrix ${\sf \Lambda}$ are in non-increasing order.

The following result establishes the connection between the physical directions
$\{ \phi_{ {\sf R},\ell}, \hsppp \phi_{ {\sf T}, \ell} \}$ in the ULA channel model
in~(\ref{eq_1}) and $\left\{ {\bf f}_{\sf opt},  \hsppp
{\bf g}_{\sf opt} \right\}$ in~(\ref{eq_2}) and confirms the observations in many
recent works~\cite{oelayach,oelayach1}.
\begin{thm}
\label{thm1}
With ${\bf H} = {\sf H}$ and the channel model in~(\ref{eq_1}), all the eigenvectors
of ${\sf H}^H{\sf H}$ can be represented as linear combinations of ${\bf v}_1, \cdots ,
{\bf v}_L$ (the transmit array steering vectors). Thus, ${\bf f}_{\sf opt}$ is a
linear combination of ${\bf v}_1, \cdots , {\bf v}_L$ and ${\bf g}_{\sf opt}$ is a
linear combination of ${\bf u}_1, \cdots , {\bf u}_L$.
\end{thm}
\begin{proof}
See Appendix~\ref{appA}.
\end{proof}

Theorem~\ref{thm1} suggests the efficacy of directional beamforming when the channels
are sparse~\cite{oelayach,zhao,akbar,vasanth_jstsp,vasanth_it2}, as is likely the
case in mmW systems. On the other hand, limited feedback schemes commonly used at
cellular frequencies~\cite{david_review,david_grass,mukkavilli} for CSI acquisition
are similar in spirit to directional beamforming schemes for mmW systems. In
particular, the typically non-sparse nature of channels at cellular frequencies
(corresponding to a large number of paths) washes away any of the underlying Fourier
structure~\cite{akbar} of the steering vectors with a uniformly-spaced array.
Without any specific structure on the space of optimal beamforming vectors, a
good limited feedback codebook such as a Grassmannian line packing solution
uniformly quantizes the space of all beamforming vectors.

On a technical note, Theorem~\ref{thm1} provides a non-unitary basis\footnote{The
use of non-unitary bases in the context of MIMO system studies is not new; see,
e.g.,~\cite{vasanth_asilomar,vasanth_arxiv2011}.} for the eigen-space of
${\sf H}^H {\sf H}$ (with eigenvalues greater than $0$)
when $L \leq N_t$ (in the $L > N_t$ case, $\{ {\bf v}_1, \cdots , {\bf v}_L \}$ span
the eigen-space but do not form a basis). The intuitive meaning of ${\bf f}_{\sf opt}$
and ${\bf g}_{\sf opt}$ is that they perform ``coherent'' beam-combining by appropriate
phase compensation to maximize the energy delivered to the receiver. As an illustration
of Theorem~\ref{thm1}, in the special case of $L = 2$ paths, when ${\bf v}_1$ and ${\bf v}_2$
are orthogonal (${\bf v}_1^H {\bf v}_2 = 0$), the non-unit-norm version of
${\bf f}_{\sf opt}$ and ${\bf g}_{\sf opt}$ are given as
\begin{eqnarray}
{\bf f}_{\sf opt} & = & \beta_{\sf opt} \cdot {\bf v}_1 +
e^{j \left(\angle{\alpha_1} - \angle{\alpha_2} - \angle{ {\bf u}_1^H {\bf u}_2 } \right) }
\sqrt{1 - \beta_{\sf opt}^2}  \cdot {\bf v}_2
\label{eq8}
\\
{\bf g}_{\sf opt} & = & \alpha_1 \beta_{\sf opt} \cdot {\bf u}_1 +
e^{j \left(\angle{\alpha_1} - \angle{\alpha_2} - \angle{ {\bf u}_1^H {\bf u}_2 } \right) }
\alpha_2 \sqrt{1 - \beta_{\sf opt}^2}  \cdot {\bf u}_2
\label{eq9}
\end{eqnarray}
where
\begin{eqnarray}
\beta_{\sf opt}^2  = 
\frac{1}{ 2} \cdot \left[ 1 +
\frac{ |\alpha_1|^2 - |\alpha_2|^2}
{ \sqrt{  \left( |\alpha_1|^2 - |\alpha_2|^2 \right)^2 + 4
|\alpha_1|^2 |\alpha_2|^2 \cdot |{\bf u}_1^H {\bf u}_2|^2 }}
\right] .
\label{eq10}
\end{eqnarray}
In addition, if ${\bf u}_1$ and ${\bf u}_2$ are orthogonal, it can be seen that
$\beta_{\sf opt}$ is either $1$ or $0$ with full power allocated to the dominant path.
At the other extreme of near-parallel ${\bf u}_1$ and ${\bf u}_2$, the optimal scheme
converges to proportional power allocation. That is,
\begin{eqnarray}
|{\bf u}_1^H {\bf u}_2| \rightarrow 1 \Longrightarrow \beta_{\sf opt}^2
\rightarrow \frac{ |\alpha_1|^2} { |\alpha_1|^2 + |\alpha_2|^2}.
\label{eq11}
\end{eqnarray}

If ${\bf u}_1$ and ${\bf u}_2$ are orthogonal (${\bf u}_1^H {\bf u}_2 = 0$),
the non-unit-norm version of ${\bf f}_{\sf opt}$ and ${\bf g}_{\sf opt}$ are given as
\begin{align}
&{\bf f}_{\sf opt} = \beta_{\sf opt} \cdot {\bf v}_1 +
e^{ - j \angle{ {\bf v}_1^{H} {\bf v}_2}} \sqrt{1 - \beta_{\sf opt}^2} \cdot
{\bf v}_2 
\label{eq12}
\\
&{\bf g}_{\sf opt} =  \alpha_1 \left( \beta_{\sf opt} + | {\bf v}_1^H {\bf v}_2 |
 \sqrt{1 - \beta_{\sf opt}^2} \right) \cdot {\bf u}_1 +
e^{ - j \angle{ {\bf v}_1^{H} {\bf v}_2}} \alpha_2
\left( | {\bf v}_1^H {\bf v}_2 | \beta_{\sf opt} +
 \sqrt{1 - \beta_{\sf opt}^2} \right) \cdot {\bf u}_2
\label{eq13}
\end{align}
where
\begin{eqnarray}
\beta_{\sf opt}^2 & = & \left\{
\begin{array}{cc}
\frac{ {\cal A} + \sqrt{ {\cal B}} }{2 {\hspace{0.01in}} {\cal C}}
& {\rm if} {\hspace{0.1in}} |\alpha_1| \geq |\alpha_2| \\
\frac{ {\cal A} - \sqrt{ {\cal B}} }{2 {\hspace{0.01in}} {\cal C}}
& {\rm if} {\hspace{0.1in}} |\alpha_1| < |\alpha_2|
\end{array}
{\hspace{0.10in}} {\rm with} \right.
\nonumber \\
{\cal A} & = & \frac{ (|\alpha_1|^2 - |\alpha_2|^2)^2 }
{ |{\bf v}_1^{H} {\bf v}_2|^2 } + 2 |\alpha_1|^2 \cdot ( |\alpha_1|^2 + |\alpha_2|^2 )
\nonumber \\
{\cal B} & = & \frac{ (|\alpha_1|^2 - |\alpha_2|^2 )^4 }
{ |{\bf v}_1^{H} {\bf v}_2|^4 } +
\frac{ 4 |\alpha_1|^2 |\alpha_2|^2 }{ | {\bf v}_1^{H} {\bf v}_2|^2 }
\cdot \left( |\alpha_1|^2 - |\alpha_2|^2 \right)^2
\nonumber \\
{\cal C} & = & \left( 1 + \frac{1}{ |{\bf v}_1^{H} {\bf v}_2|^2} \right)
\cdot \left( |\alpha_1|^2 + |\alpha_2|^2  \right)^2 -
\frac{ 4 |\alpha_1|^2 |\alpha_2|^2 }{ |{\bf v}_1^{H} {\bf v}_2 |^2 }.
\label{eq14}
\end{eqnarray}
For specific examples, note that $\beta_{\sf opt}$ converges to $1$ or $0$ as
${\bf v}_1$ and ${\bf v}_2$ become more orthogonal. On the other hand, the optimal
scheme converges to proportional squared power allocation as ${\bf v}_1$ and
${\bf v}_2$ become more parallel. That is,
\begin{eqnarray}
|{\bf v}_1^{H} {\bf v}_2| \rightarrow 1 \Longrightarrow
\beta_{\sf opt}^2 \rightarrow \frac{ |\alpha_1|^4 }
{ |\alpha_1|^4 + |\alpha_2|^4 }.  
\label{eq15}
\end{eqnarray}
Similar expressions can be found in~\cite{vasanth_gcom15} for the cases where ${\bf v}_1$
and ${\bf v}_2$ are near-parallel (${\bf v}_1^H {\bf v}_2 \approx 1$), ${\bf u}_1$
and ${\bf u}_2$ are near-parallel (${\bf u}_1^H {\bf u}_2 \approx 1$), etc.

\subsection{Optimal Beamforming with Phase-Only Control}
\label{sec3_b}

In practice, the antenna arrays at the MWB and UE ends are often controlled by a
common PA disallowing per-antenna power control. Thus, there is a need to
understand the performance with phase-only control at both the ends. Let
${\cal F}_{\infty}^{N_t}$ denote the class of amplitude-constrained beamforming
vectors with phase-only control reflecting such an assumption. That is,
${\cal F}_{\infty}^{N_t} \triangleq \{ {\bf f} \in {\mathbb{C}}^{N_t} \hsppp :
\hsppp \| {\bf f} \|_{\infty} \leq \frac{1}{ \sqrt{N_t} } \}$. We now consider
the problem of optimal beamforming with ${\bf f} \in {\cal F}_{\infty}^{N_t}$
and ${\bf g} \in {\cal F}_{\infty}^{N_r}$. Note that if ${\bf f} \in
{\cal F}_{\infty}^{N_t}$, then ${\bf f} \in {\cal F}_{2}^{N_t}$. However, unlike
the optimization over ${\cal F}_{2}^{N_t}$, it is not clear that the received
${\sf SNR}$ is maximized by a choice ${\bf f}_{\sf opt}$ with
$\| {\bf f}_{\sf opt} \|_{\infty} = \frac{1}{ \sqrt{N_t} }$. Nor is it clear that
$\| {\bf f}_{\sf opt} \|_2 = 1$. With this background, we have the following result.
\begin{thm}
\label{thm2}
The optimal choice ${\bf f}_{\sf opt}$ from ${\cal F}_{\infty}^{N_t}$ is an equal
gain transmission scheme. That is, $| {\bf f}_{\sf opt}(i) | = \frac{1}{ \sqrt{N_t} },
\hspp i = 1, \cdots, N_t$. The optimal choice ${\bf g}_{\sf opt}$ from ${\cal F}_{\infty}^{N_r}$
satisfies
\begin{eqnarray}
{\bf g}_{\sf opt} = \frac{1}{ \sqrt{N_r} } \cdot  \frac{ {\sf H} \hsppp {\bf f}_{\sf opt} }
{ \| {\sf H} \hsppp {\bf f}_{\sf opt} \|_{\infty} }.
\label{eq_4}
\end{eqnarray}
\end{thm}
\begin{proof}
See Appendix~\ref{appB}.
\end{proof}

Note that as with the proof of the optimal beamforming structure from
${\cal F}_2^{N_t}$ and ${\cal F}_2^{N_r}$ where given a fixed ${\bf f} \in
{\cal F}_2^{N_t}$, the matched filter combiner corresponding to it is optimal
from ${\cal F}_2^{N_r}$, the matched filter combiner structure in~(\ref{eq_4}) is
optimal for any fixed ${\bf f} \in {\cal F}_{\infty}^{N_t}$. Further, (by
construction), ${\bf f}_{\sf opt}$ corresponds to an equal gain transmission
scheme, which is also power efficient. That is, $\| {\bf f}_{\sf opt} \|_2 = 1$
and $\| {\bf f}_{\sf opt} \|_{\infty} = \frac{ 1 }{ \sqrt{N_t} }$. On the other
hand, while $\| {\bf g}_{\sf opt} \|_{\infty} = \frac{1} { \sqrt{N_r} }$, it is
unclear if ${\bf g}_{\sf opt}$ corresponds to an equal gain reception scheme, let
alone a power efficient scheme. That is, not only can $\| {\bf g}_{\sf opt} \|_2$
be smaller than $1$, but also $| {\bf g}_{\sf opt}(i)|$ need not be $\frac{1}
{ \sqrt{N_r} }$ for some $i$.

While Theorem~\ref{thm2} specifies the amplitudes of $\left\{ {\bf f}_{\sf opt}(i) \right\}$,
it is unclear on the phases of $\left\{ {\bf f}_{\sf opt}(i) \right\}$. In general,
the search for the optimal phases of ${\bf f}_{\sf opt}(i)$ appears to be a quadratic
programming problem with attendant issues on initialization and convergence to local
optima. We now provide two good solutions (as evidenced by their performance
relative to the optimal scheme from Sec.~\ref{sec3_a} in subsequent numerical studies)
to the received ${\sf SNR}$ maximization problem from ${\cal F}_{\infty}^{N_t}$ and
${\cal F}_{\infty}^{N_r}$. The first solution is the equal-gain RSV and its matched
filter as candidate beamforming vectors at the two ends (${\bf f}_{\sf cand, \hsppp 1}$
and ${\bf g}_{\sf cand, \hsppp 1}$):
\begin{eqnarray}
{\bf f}_{\sf cand, \hsppp 1}(i) = \frac{1}{ \sqrt{N_t} } \cdot \angle{ {\sf v}_1(i) },
\hspp
{\bf g}_{\sf cand, \hsppp 1} = \frac{1}{ \sqrt{N_r} } \cdot \frac{ {\sf H}
\hsppp {\bf f}_{\sf cand, \hsppp 1}}
{ \| {\sf H} \hsppp {\bf f}_{\sf cand, \hsppp 1} \|_{\infty} }.
\label{eq_5}
\end{eqnarray}
For the second solution, we have the following statement.
\begin{prop}
Let ${\sf H}$ be decomposed along the column vectors as ${\sf H} = \left[
{\sf h}_1, \cdots, {\sf h}_{N_t} \right]$. With $\theta_1 = 0$, let $\theta_i$
be recursively defined as
\begin{eqnarray}
\theta_i = \angle{ \left( \sum_{k = 1}^{i-1} e^{j \theta_k} \cdot
{\sf h}_i^H {\sf h}_k \right).}
\label{eq18}
\end{eqnarray}
The beamforming vector ${\bf f}_{\sf cand, \hsppp 2}$ where ${\bf f}_{\sf cand, \hsppp 2}(i)
= \frac{1} { \sqrt{N_t} } \cdot e^{j \theta_i}$ and ${\bf g}_{\sf cand,\hsppp 2} =
\frac{1} { \sqrt{N_t} } \cdot \frac{ {\sf H} \hsppp {\bf f}_{\sf cand, \hsppp 2} }
{ \| {\sf H} \hsppp {\bf f}_{\sf cand, \hsppp 2} \|_{\infty} }$ lead to a good
beamforming solution for the problem considered in this section.
\label{prop1}
\end{prop}
\begin{proof}
See Appendix~\ref{appC}.
\end{proof}

The importance of the beamforming structure from Prop.~\ref{prop1} relative to the
one in~(\ref{eq_5}) is that while the latter is just a equal-gain quantization
of~(\ref{eq_2}) and thus requires the entire ${\sf H}$ for its design, the former
is more transparent in terms of the column vectors of ${\sf H}$ and can thus be
designed via a simple uplink training scheme.

Fig.~\ref{fig1}(a) plots the complementary cumulative distribution function (complementary
CDF) of the loss in ${\sf SNR}_{\sf rx}$ between the optimal beamforming scheme in~(\ref{eq_2})
and four candidate beamforming schemes: i) equal-gain RSV and matched filter scheme
from~(\ref{eq_5}), ii) beamforming scheme from Prop.~\ref{prop1}, iii) beamforming
along the dominant direction at the MWB and the matched filter to the dominant
direction at the UE, and iv) beamforming along the dominant directions at the MWB
and the UE in a $N_r = 4, N_t = 64$ system generated by $L = 2$ clusters whose AoAs/AoDs
are independently and identically distributed (i.i.d.) in the $120^{\sf o}$ field-of-view
(coverage area) of the arrays in the azimuth. From this study, we see that the
performance of the scheme from Prop.~\ref{prop1} is similar to that from~(\ref{eq_5}),
as is the replacement of matched filter at the UE end with the dominant direction.

\begin{figure*}[htb!]
\begin{center}
\begin{tabular}{cc}
\includegraphics[height=2.7in,width=3.15in] {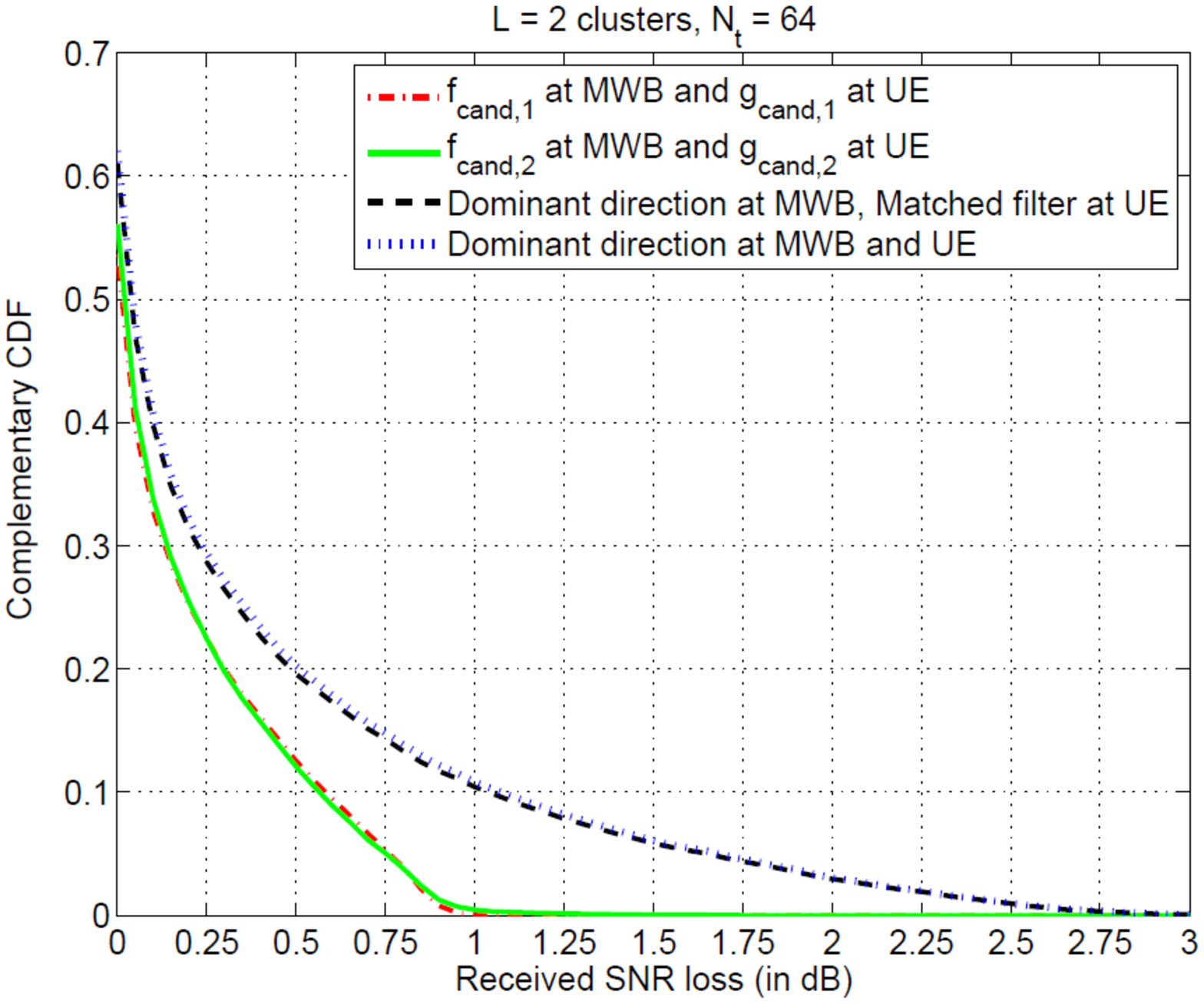}
&
\includegraphics[height=2.7in,width=3.15in] {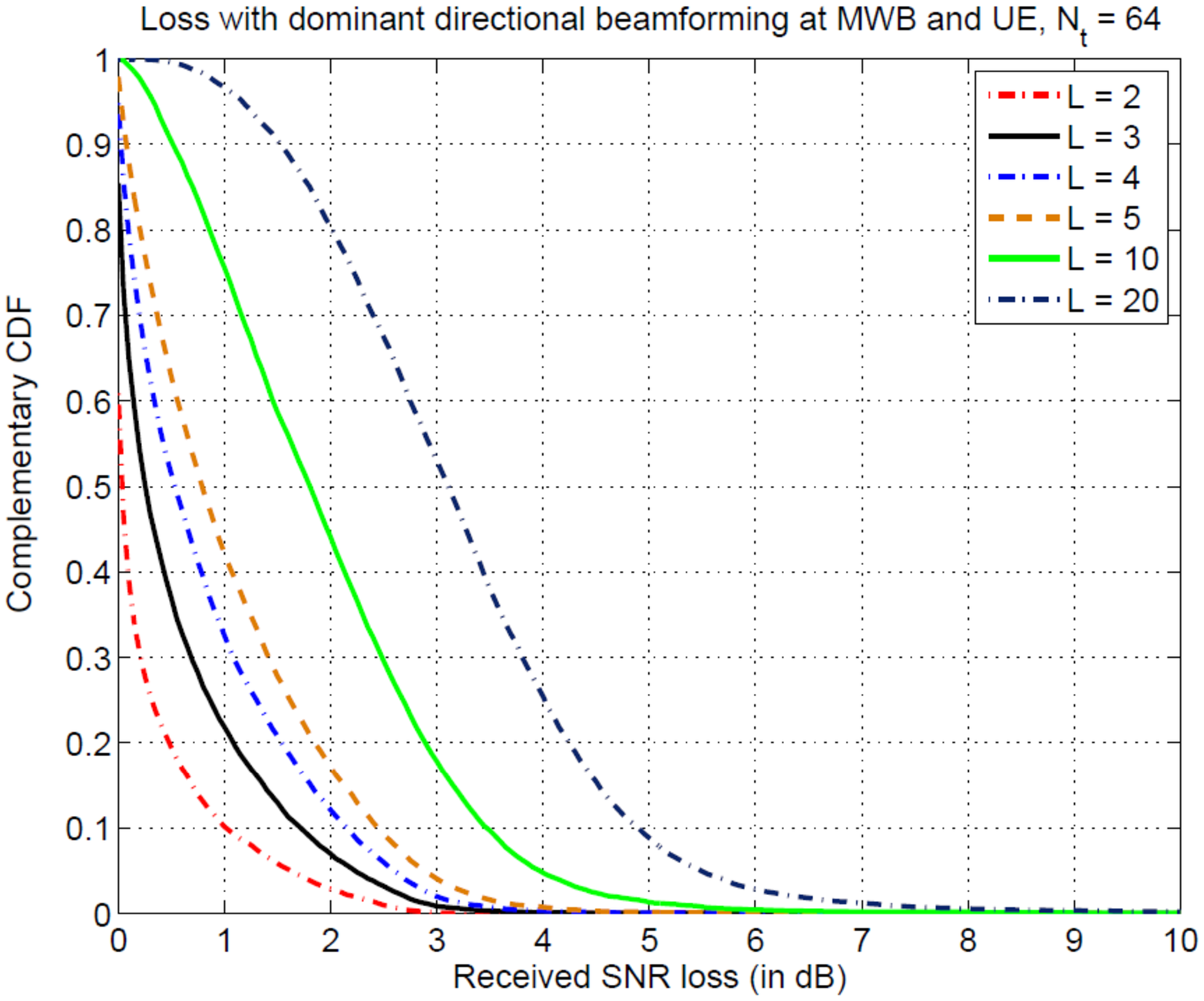}
\\ {\hspace{0.35in}} (a) & {\hspace{0.12in}} (b)
\end{tabular}
\caption{\label{fig1}
Complementary CDF of loss in ${\sf SNR}_{\sf rx}$ between the optimal beamforming
scheme from ${\cal F}_2^{N_t}$ in~(\ref{eq_2}) relative to: (a) different
beamforming solutions with $L = 2$ clusters, and (b) a dominant directional
beamforming scheme with different choices of $L$. }
\end{center}
\vspace{-5mm}
\end{figure*}

\subsection{Issues with RSV Learning}
\label{sec_3c}

The (near-)optimality of the RSV-type solutions from ${\cal F}_2^{N_t}$ and
${\cal F}_{\infty}^{N_t}$ suggests that a reasonable approach for beamformer design
is to let the MWB and UE learn an approximation to ${\bf f}_{\sf opt}$ and
${\bf g}_{\sf opt}$, respectively. A similar approach is adopted at cellular
frequencies under the rubric of limited feedback schemes that approximate the RSV
of the channel from a codebook of beamforming vectors. We specialize this approach
and elaborate on their appropriateness for mmW systems.

A well-known RSV learning scheme that exploits the time division duplexing(TDD)-reciprocity
of the channel ${\sf H}$ is power iteration~\cite[Sec.\ 7.3]{golub} which iterates
($i = 0, 1, \cdots$) a randomly initialized beamforming vector (${\bf f}_0$) over the
channel as follows:
\begin{eqnarray}
\widetilde{\bf g}_{i+1} & = & \sqrt{ \rho_{\sf f} } \cdot {\sf H} {\bf f}_i +
{\bf n}_{ {\sf f}, \hsppp i +1}
\label{eq_7}
\\
{\bf g}_{i+1} & = & \frac{ \widetilde{\bf g}_{i+1} } { \| \widetilde{\bf g}_{i+1} \|_2}
= \frac{ \sqrt{ \rho_{\sf f} } \cdot {\sf H}  {\bf f}_i + {\bf n}_{ {\sf f}, \hsppp i +1} }
{ \| \sqrt{ \rho_{\sf f} } \cdot {\sf H} {\bf f}_i + {\bf n}_{ {\sf f}, \hsppp i +1} \|_2 }
\label{eq_8}
\\
{\bf z}_{i+1} & = & \sqrt{ \rho_{\sf r} } \cdot {\sf H}^T {\bf g}_{i + 1}^{\star} +
{\bf n}_{ {\sf r}, \hsppp i+1}^{\star}
\label{eq_9}
\\
{\bf f}_{i+1} & = & \frac{ {\bf z}_{i+1}^{\star} } { \| {\bf z}_{i+1} \|_2 }
= \frac{ \sqrt{ \rho_{\sf r} } \cdot {\sf H}^H {\bf g}_{i + 1} + {\bf n}_{ {\sf r}, \hsppp i+1} }
{ \| \sqrt{ \rho_{\sf r} } \cdot {\sf H}^H {\bf g}_{i + 1} + {\bf n}_{ {\sf r}, \hsppp i+1} \|_2 }.
\label{eq_10}
\end{eqnarray}
In~(\ref{eq_7})-(\ref{eq_10}), $\rho_{\sf f}$ and $\rho_{\sf r}$ stand for the
pre-beamforming ${\sf SNR}$s of the forward and reverse links, respectively. A
straightforward simplification shows that
\begin{eqnarray}
{\bf f}_{i+1} = \frac{
\sqrt{ \rho_{\sf f} \rho_{\sf r} } \cdot {\sf H}^H {\sf H} \hsppp {\bf f}_i +
\sqrt{\rho_{\sf r} } \hsppp {\sf H}^H {\bf n}_{ {\sf f}, \hsppp i + 1} +
\| \sqrt{ \rho_{\sf f}} \hsppp {\sf H} \hsppp {\bf f}_i + {\bf n}_{ {\sf f}, \hsppp i + 1} \|_2
\cdot {\bf n}_{ {\sf r}, \hsppp i + 1} }
{ \| \sqrt{ \rho_{\sf f} \rho_{\sf r} } \cdot {\sf H}^H {\sf H} \hsppp {\bf f}_i +
\sqrt{\rho_{\sf r} } \hsppp {\sf H}^H {\bf n}_{ {\sf f}, \hsppp i + 1} +
\| \sqrt{ \rho_{\sf f}} \hsppp {\sf H} \hsppp {\bf f}_i + {\bf n}_{ {\sf f}, \hsppp i + 1} \|_2
\cdot {\bf n}_{ {\sf r}, \hsppp i + 1} \|_2 }.
\label{eq23}
\end{eqnarray}
When the system is noise-free ($\{ \rho_{\sf f} , \hsppp \rho_{\sf r} \} \rightarrow \infty$),
the above algorithm reduces to
\begin{eqnarray}
{\bf f}_{i+1} & \rightarrow & \frac{ {\sf H}^H {\sf H} \hsppp {\bf f}_i }
{ \| {\sf H}^H {\sf H} \hsppp {\bf f}_i \|_2 }  \Longrightarrow {\bf f}_{i+1}
\rightarrow \frac{ \left( {\sf H}^H {\sf H} \right)^{i+1} {\bf f}_0 }
{ \| \left( {\sf H}^H {\sf H} \right)^{i+1} {\bf f}_0  \|_2}
\label{eq24}
\\
{\bf g}_{i+1} & \rightarrow & \frac{ {\sf H} \hsppp {\bf f}_i }
{ \| {\sf H} \hsppp {\bf f}_i \|_2 }
\Longrightarrow {\bf g}_{i+1} \rightarrow
 \frac{ {\sf H} \left( {\sf H}^H {\sf H} \right)^{i+1} {\bf f}_0 }
{ \| {\sf H} \left( {\sf H}^H {\sf H} \right)^{i+1} {\bf f}_0  \|_2} .
\label{eq25}
\end{eqnarray}
With ${\sf H}^H {\sf H} = {\sf V} {\sf \Lambda} {\sf V}^H$, we have
\begin{eqnarray}
{\bf f}_{i+1} = \frac{ \sum_{j = 1}^{N_t} {\sf v}_j \cdot
\left( {\sf v}_j^H {\bf f}_0 \right) \cdot \left( {\sf \Lambda}_j \right)^{i+1} }
{ \|\sum_{j = 1}^{N_t} {\sf v}_j \cdot
\left( {\sf v}_j^H {\bf f}_0 \right) \cdot \left( {\sf \Lambda}_j \right)^{i+1} \|_2 }.
\label{eq26}
\end{eqnarray}

\begin{figure*}[htb!]
\begin{center}
\begin{tabular}{cc}
\includegraphics[height=2.7in,width=3.15in] {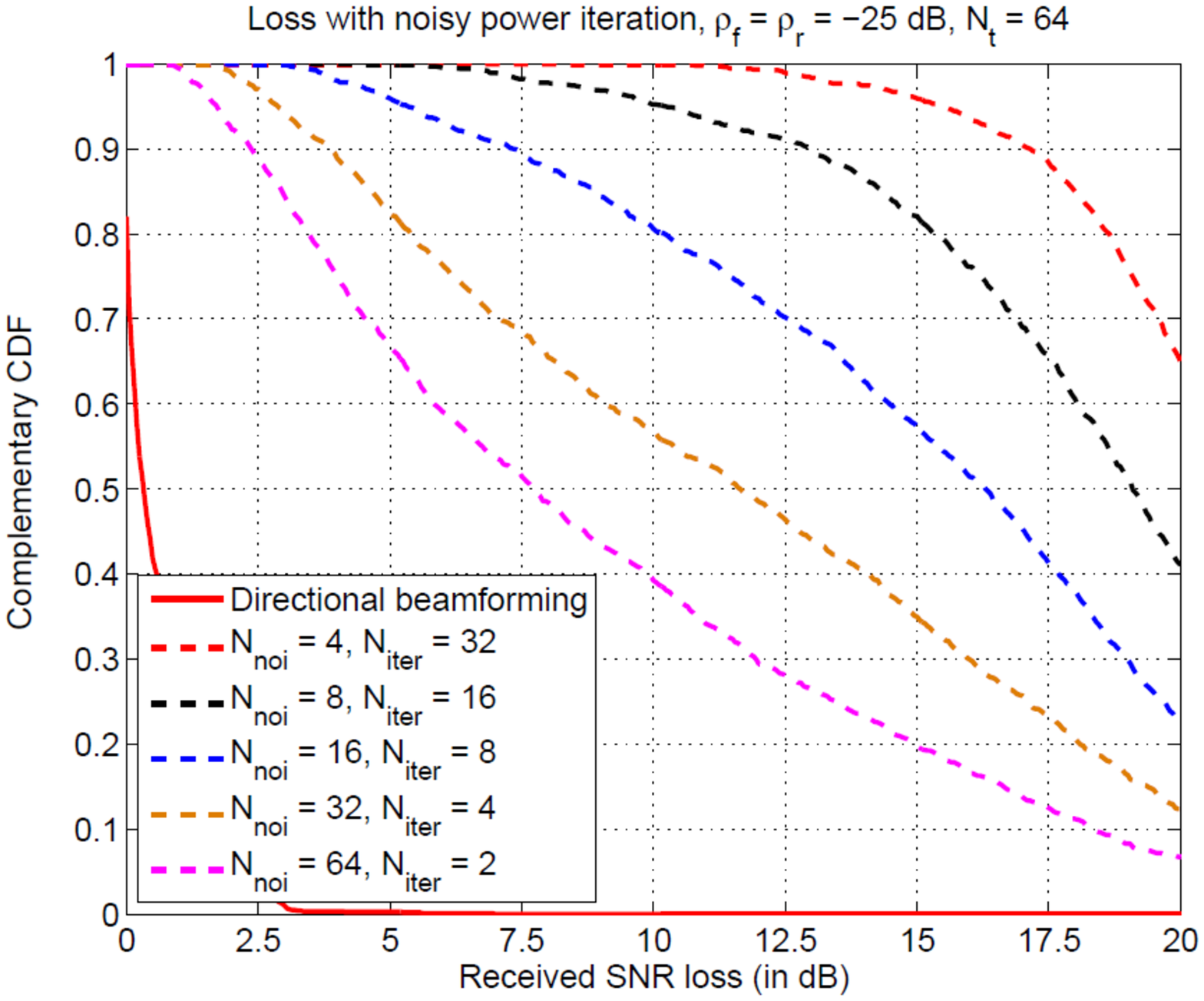}
&
\includegraphics[height=2.7in,width=3.15in] {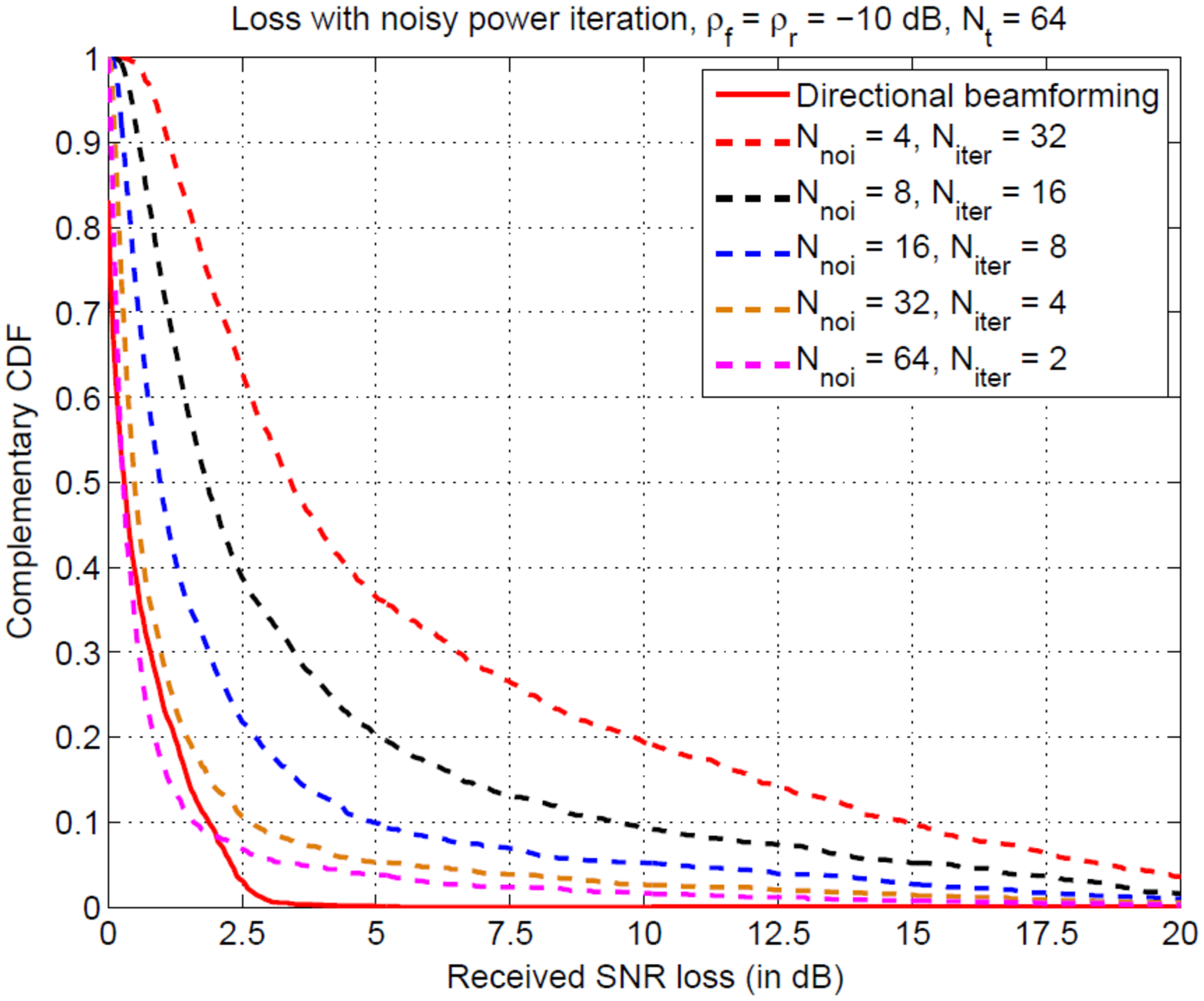}
\\ {\hspace{0.35in}} (a) & {\hspace{0.12in}} (b)
\end{tabular}
\caption{\label{fig2}
Complementary CDF of loss in ${\sf SNR}_{\sf rx}$ between the optimal beamforming
scheme from ${\cal F}_2^{N_t}$ in~(\ref{eq_2}) and the noisy power iteration scheme
with (a) $\rho_{\sf f} = \rho_{\sf r} = -25$ dB and (b) $\rho_{\sf f} = \rho_{\sf r}
= -10$ dB.}
\end{center}
\vspace{-5mm}
\end{figure*}

While the noise-free power iteration scheme has been proposed in the MIMO context
in~\cite{dahl}, understanding the performance tradeoff of the noisy case analytically
appears to be a difficult problem, in general. To surmount this difficulty, we
numerically study the performance of the noisy case at a low pre-beamforming ${\sf SNR}$
on the order of $-25$ to $-10$ dB. We consider the case where $N_{\sf npi} = 256$
samples are used for RSV learning and these samples are partitioned in different
ways\footnote{Note that the $2$ factor in the partition of $N_{\sf npi}$ arises because
power iteration is bi-directional.} as $N_{\sf npi} = 2N_{\sf noi} \times N_{\sf iter}$.
Here, $N_{\sf noi}$ samples are used to improve $\rho_{\sf f}$ by noise averaging and
$N_{\sf iter}$ samples are used for beamformer iteration. In particular, we consider the
following choices for $N_{\sf noi}$ in our study: $\{4, 8, 16, 32, 64 \}$ with $\rho_{\sf f}
= \rho_{\sf r} = \{-25, \hsppp -10 \}$ dB and Figs.~\ref{fig2}(a)-(b) plot the
complementary cumulative distribution function (CDF)
of the loss in ${\sf SNR}_{\sf rx}$ for these two scenarios in a $L = 2, N_r = 4$ and
$N_t = 64$ system with averaging over the random choice of ${\bf f}_0$. From these two
plots, we observe that given $N_{\sf npi}$ samples, noise averaging is a task of higher
importance at low pre-beamforming ${\sf SNR}$s than beamformer iteration. Nevertheless,
in spite of the best noise averaging, the noisy power iteration scheme has a poor
performance at low ${\sf SNR}$s (for a large fraction of the users at $-25$ dB and a good
fraction at $-10$ dB) as noise is amplified in the iteration process rather than the
channel's RSV.

The RSV learning scheme also suffers from other problems that make its applicability
in mmW systems difficult. Since each user's RSV has to be learned via a bi-directional
iteration, it is not amenable (in this form) as a common broadcast solution for the
downlink setting. This is particularly disadvantageous and impractical if each MWB has
to initiate a unicast session with a UE that is yet to be discovered. Further, the need
to sample each antenna individually (at both ends) can considerably slow down the iteration
process with RF hardware constraints (e.g., when there are fewer RF chains than antennas).
In addition, this approach requires calibration of the receive-side RF chain relative to
the transmit-side RF chain with respect to phase and amplitude as well as phase coherence
during the iteration. More importantly, this approach critically depends on TDD
reciprocity, which could be complicated in certain deployment scenarios that do not allow
this possibility~\cite{roh}.

An alternate approach given the RSV structure in Theorem~\ref{thm1} is to learn the
dominant directions at the MWB end $\{ \widehat{\phi}_{ {\sf T}, \ell} \}$ and then
combine the beams with appropriate weights $\{ \widehat{\alpha}_{\ell} \}$ to result
in a 
beamforming vector:
\begin{eqnarray}
{\bf f}_{\sf comb} = \frac{  \sum_{\ell = 1}^L \widehat{\alpha}_{\ell} \hspp
{\sf CPO} ( \widehat{ \phi}_{ {\sf T}, \ell } ) }
{\| \sum_{\ell = 1}^L \widehat{\alpha}_{\ell} \hspp
{\sf CPO} ( \widehat{ \phi}_{ {\sf T}, \ell } )\|_2 }.
\label{eq27}
\end{eqnarray}
The difficulty with this approach is that it suffers from PA inefficiency (not all
the PAs operate at maximal power). Fig.~\ref{fig3}(a) plots the complementary CDF of
the peak-to-average ratio 
(PAR) of ${\bf f}_{\sf comb}$, defined as,
\begin{eqnarray}
{\sf PAR} \triangleq  
\frac{ \max_i | {\bf f}_{\sf comb}(i) |^2}{ 1/ N_t} =
N_t \cdot \max_i | {\bf f}_{\sf comb}(i)|^2
\end{eqnarray}
corresponding to beam combining along two randomly chosen, but known directions
with random weights. Fig.~\ref{fig3}(a) shows that a median PAR loss of over $2$ dB
is seen for $N_t \geq 8$ suggesting that the RSV gain relative to directional
beamforming of less than a dB (see Fig.~\ref{fig1}(a)) is significantly outweighed
by the PA inefficiency. In other words, the ${\sf SNR}_{\sf rx}$ loss with just
selecting the dominant direction at the MWB and the UE ends is far less than the PA
backoff due to combining multiple directions at the MWB or the UE.

\begin{figure*}[htb!]
\begin{center}
\begin{tabular}{cc}
\includegraphics[height=2.7in,width=3.15in] {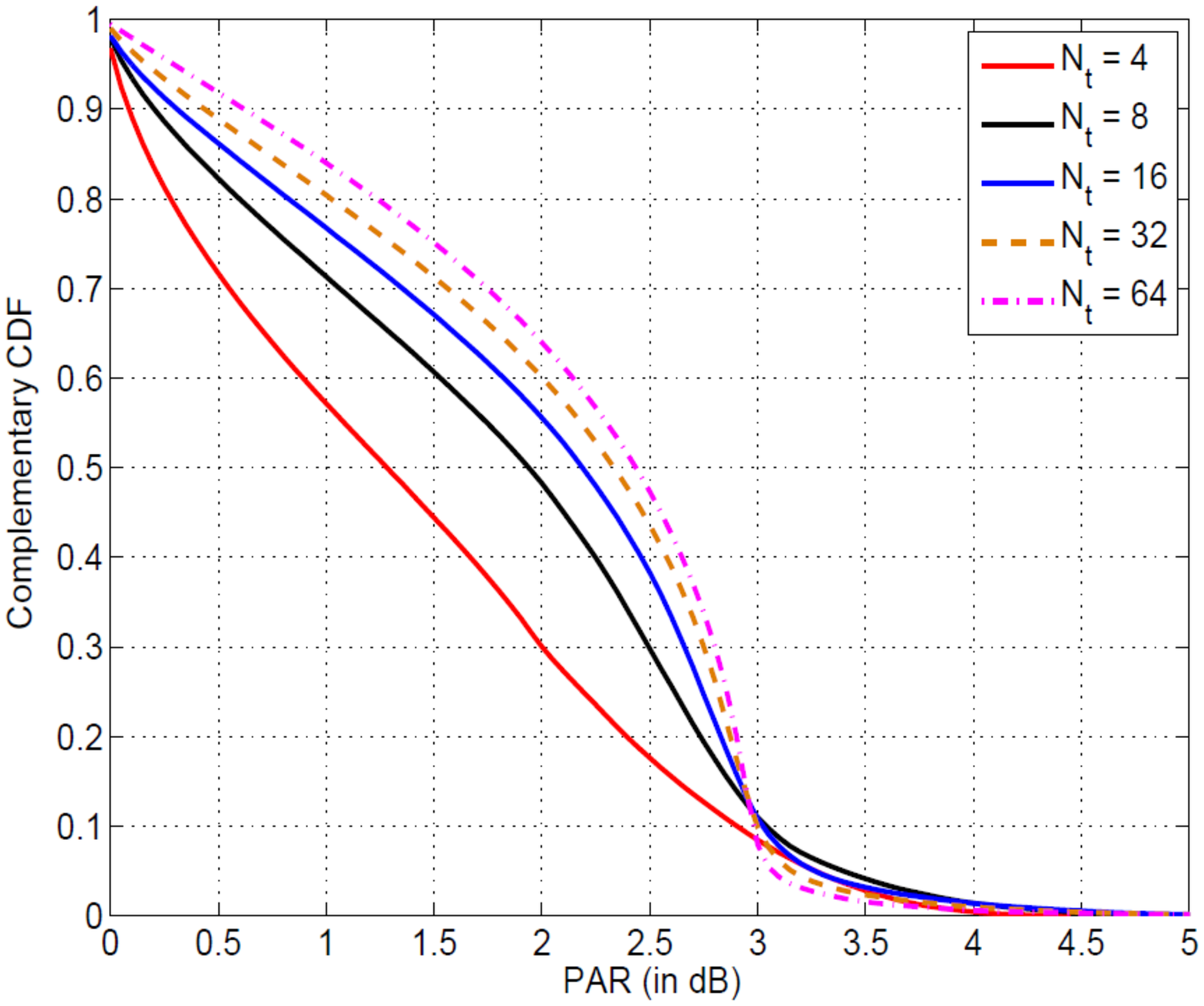}
&
\includegraphics[height=2.7in,width=3.15in] {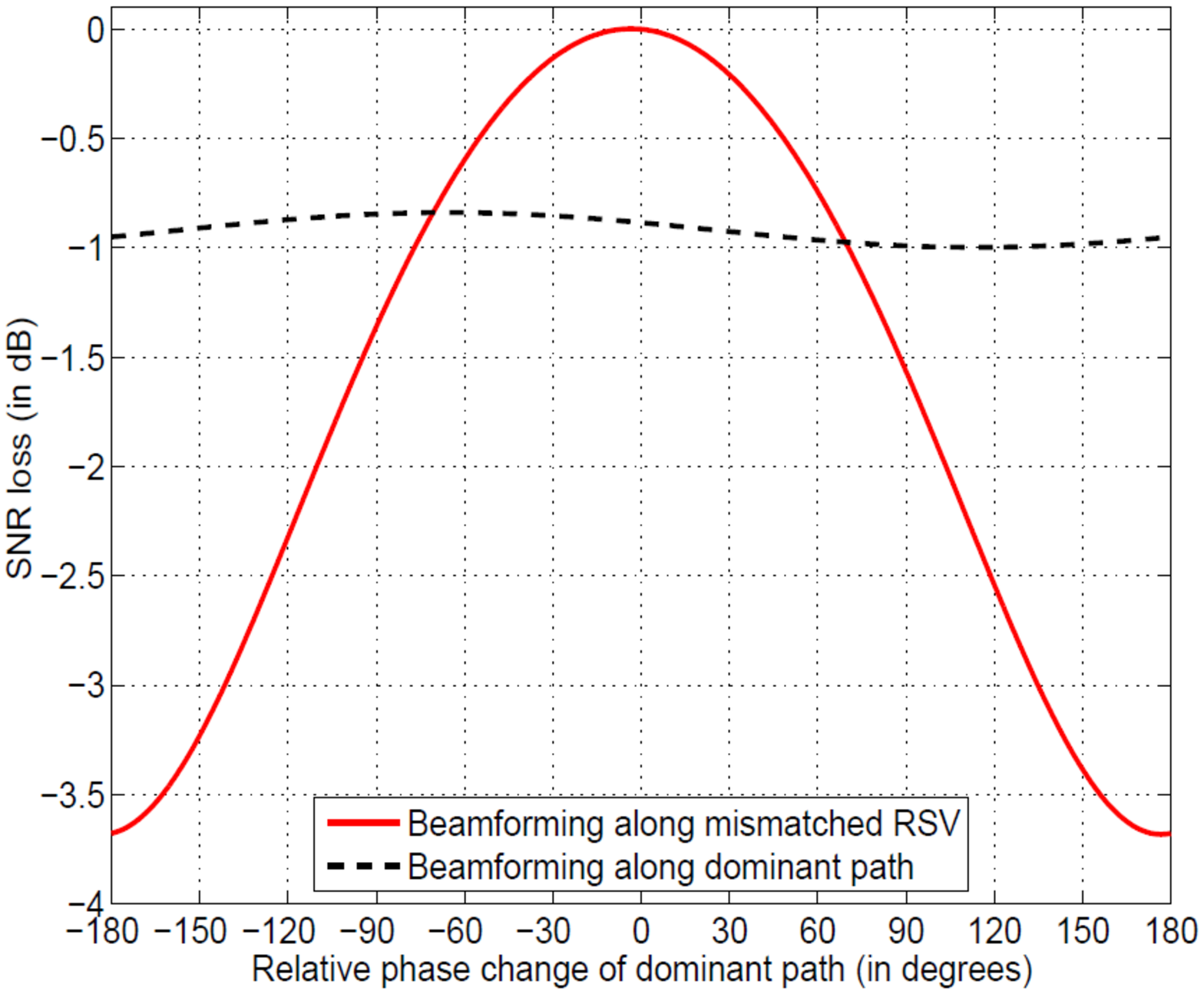}
\\ {\hspace{0.35in}} (a) & {\hspace{0.12in}} (b)
\end{tabular}
\caption{\label{fig3}
(a) Complementary CDF of PAR of the combined beamforming vector needed to mimic the
RSV structure. (b) Typical performance loss with mismatched RSV and mismatched
dominant directional beamforming schemes relative to optimal beamforming in the
perturbed case.}
\end{center}
\vspace{-5mm}
\end{figure*}

More generally, any RSV learning scheme is bound to be extremely sensitive\footnote{Note
that the higher sensitivity of the eigenvectors of a MIMO channel matrix (relative to
the eigenvalues) to small perturbations in the channel entries has been
well-understood~\cite{vasanth_isit05,vasanth_jsac}. See~\cite[Sec.\ 7.2]{golub} for
theoretical details.} to relative phase changes across paths. For example, Fig.~\ref{fig3}(b)
plots the typical behavior of loss in ${\sf SNR}_{\sf rx}$ with the mismatched reuse of
the optimal beamformer and the dominant directional beamformer (both from the unperturbed
case) relative to the optimal beamformer in the perturbed case as the phase of the dominant
path in a $L = 2, N_r = 4, N_t = 64$ system changes. In this example, the two paths are
such that $\phi_{ {\sf R},1} = 108.57^{\sf o}$, $\phi_{ {\sf T},1} = 83.74^{\sf o}$,
$\alpha_1 = 2.61$, $\phi_{ {\sf R},2} = 92.74^{\sf o}$, $\phi_{ {\sf T},2} = 94.26^{\sf o}$,
$\alpha_2 = 1.79$. We note that the RSV scheme takes a steep fall in performance as the
phase changes, whereas the directional scheme remains approximately stable in performance.
It is important to note that a $360^{\sf o}$ change in phase corresponds to a change in path
length of $\lambda$ (a small distance at higher carrier frequencies and hence an increasingly
likely possibility). Such a sensitivity for any RSV reconstruction scheme to phase changes
renders this approach's utility in the mmW context questionable.

\section{Directional Beamforming and Direction Learning}
\label{sec4}

Instead of the RSV solution, we now consider the performance loss with a
low-complexity strategy that beamforms along the dominant direction at the MWB and
the UE. From the numerical study in Fig.~\ref{fig1}(a), we see that the dominant
directional beamforming scheme suffers only a minimal loss relative to even the best
scheme from ${\cal F}_2^{N_t}$ and ${\cal F}_2^{N_r}$ (a median loss of a fraction
of a dB and less than a dB even at the $90$-th percentile level). Further,
Fig.~\ref{fig1}(b) plots the complementary CDF of the loss in ${\sf SNR}_{\sf rx}$
between the optimal scheme in~(\ref{eq_2}) and the dominant directional beamforming
scheme with different choices of $L$: $L = 2, 3, 4, 5, 10$ or $20$. From this study,
we note that directional beamforming results in less than a dB loss for over $50\%$
of the users for even up to $L = 5$ clusters. Further, directional beamforming results
in no more than $2.5$ dB loss for even up to $90\%$ of the users. Thus, this study
suggests that learning the directions (AoAs/AoDs) along which the UE and MWB should
beamform is a useful strategy for initial UE discovery.

\subsection{Learning Dominant Directions via Subspace Methods}
\label{sec4_a}

AoA/AoD learning with multiple antenna arrays has a long and illustrious history
in the signal/array processing literature~\cite{krim}. In the simplest case of
estimating a single unknown source (signal direction) at the UE end with system
equation:
\begin{eqnarray}
{\bf y} = \alpha_1 {\bf u}( \phi_1) + {\bf n}
\label{eq28}
\end{eqnarray}
where $\alpha_1$ is known, ${\bf u}(\cdot)$ denotes the array steering vector
and ${\bf n} \sim {\cal CN}( {\bf 0}, {\bf I})$, it can be seen that the density
function $\log \left( f( {\bf y} | \alpha_1, \phi_1) \right)$ can be written as ${\sf C} -
( {\bf y} - \alpha_1 {\bf u}(\phi_1) )^H ( {\bf y} - \alpha_1 {\bf u}(\phi_1) )$
for an appropriately defined constant ${\sf C}$. Thus, the maximum likelihood (ML) solution
($\widehat{\phi}_1$) that maximizes the density can be seen to be $\widehat{\phi}_1 =
\arg \max_{\phi} | {\bf u}(\phi)^H {\bf y} |^2$. Rephrasing, correlation of the
received vector ${\bf y}$ for the best signal strength results in the ML solution
for the problem of signal coming from one unknown direction.

In general, if there are multiple ($K$) sources with system equation
\begin{eqnarray}
{\bf y} = \sum_{k = 1}^K \alpha_k {\bf u}(\phi_k) + {\bf n},
\label{eq29}
\end{eqnarray}
the density function of ${\bf y}$ is non-convex in the parameters resulting in a
numerical multi-dimensional search in the parameter space. In this context, the main
premise behind the 
MUSIC algorithm~\cite{schmidt} is
that the signal subspace is $K$-dimensional and is orthogonal to the noise subspace.
Furthermore, the $K$ largest eigenvalues of the estimated received covariance matrix,
$\widehat{\bf R}_{\bf y}$, correspond to the signal subspace and the other eigenvalues
to the noise subspace (provided the covariance matrix estimate is reliable). The MUSIC
algorithm then estimates the signal directions by finding the ($K$) peaks of the
pseudospectrum\footnote{In general, the choice of $K$ in~(\ref{eq_17}) has to be
estimated via an information theoretic criterion as in~\cite{wax_kailath} or via
minimum description length criteria such as those due to Rissanen or Schwartz.},
defined as,
\begin{eqnarray}
P_{\sf MUSIC}(\phi) \triangleq \frac{1}{ \sum_{n = K+1}^{N_r} | {\bf u}(\phi)^H
\widehat{\bf q}_n|^2 }
\label{eq_17}
\end{eqnarray}
where $\{ \widehat{\bf q}_{K+1}, \cdots, \widehat{\bf q}_{N_r} \}$ denote the
eigenvectors of the noise subspace of $\widehat{\bf R}_{\bf y}$. The principal advantage
of the MUSIC algorithm is that the signal maximization task has been recasted as a noise
minimization task, a one-dimensional line search problem albeit at the cost of computing
the eigenvectors of $\widehat{\bf R}_{\bf y}$. Nevertheless, since $\{ \widehat{\bf q}_1,
\cdots, \widehat{\bf q}_{N_r}\}$ can be chosen to form a unitary basis, it is seen that
MUSIC attempts to maximize $\sum_{n = 1}^K | {\bf u}(\phi)^H \widehat{\bf q}_n|^2$ (or in
other words, it assigns equal weights to all the components of the signal subspace and is
hence not ML-optimal).

We now apply the MUSIC algorithm to direction learning at the MWB and UE by a bi-directional
approach where the MWB learns the AoD by estimating the uplink covariance matrix (where the
UE trains the MWB), and the UE learns the AoA by estimating the downlink covariance matrix
(where the MWB trains the UE). We consider the case where $N_{\sf music} = 256$ samples
are used for direction learning. Since the MWB is equipped with more antennas than the UE,
we partition $N_{\sf music}$ into $N_{\sf up} = 192$ samples for uplink (AoD) training and
$N_{\sf down} = 64$ samples for downlink (AoA) training. As before, we partition $N_{\sf up}$
in different ways as $N_{\sf up} = N_{\sf up,cov} \times N_{\sf up,noi}$ where $N_{\sf up,noi}$
samples are used for link margin improvement and $N_{\sf up,cov}$ samples are used for uplink
covariance matrix estimation. In particular, we study the following choices here:
$N_{\sf up,cov} = \{12, 24, 32, 48, 64, 96 \}$. $N_{\sf down} = 64$ is partitioned as
$N_{\sf down,noi} = N_{\sf down,cov} = 8$. Figs.~\ref{fig4}(a) and (b) plot the complementary
CDF of ${\sf SNR}_{\sf rx}$ with such a bi-directional MUSIC algorithm for $\rho_{\sf f}
= \rho_{\sf r} = -25$ dB and $-10$ dB, respectively. In general, $\rho_{\sf f} > \rho_{\sf r}$
and Fig.~\ref{fig4} serves as a more optimistic characterization of the MUSIC scheme for mmW
systems.

\begin{figure*}[htb!]
\begin{center}
\begin{tabular}{cc}
\includegraphics[height=2.7in,width=3.15in] {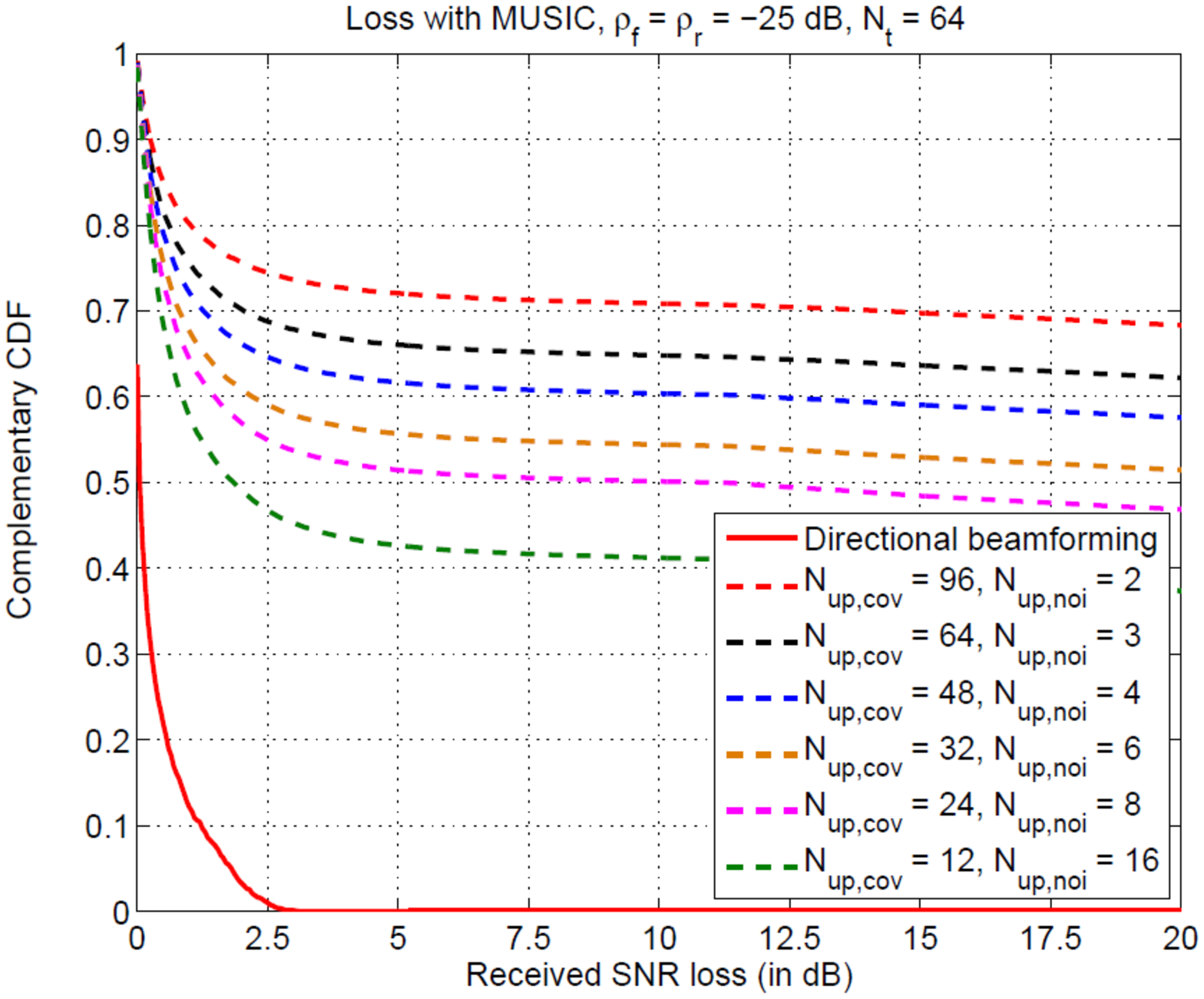}
&
\includegraphics[height=2.7in,width=3.15in] {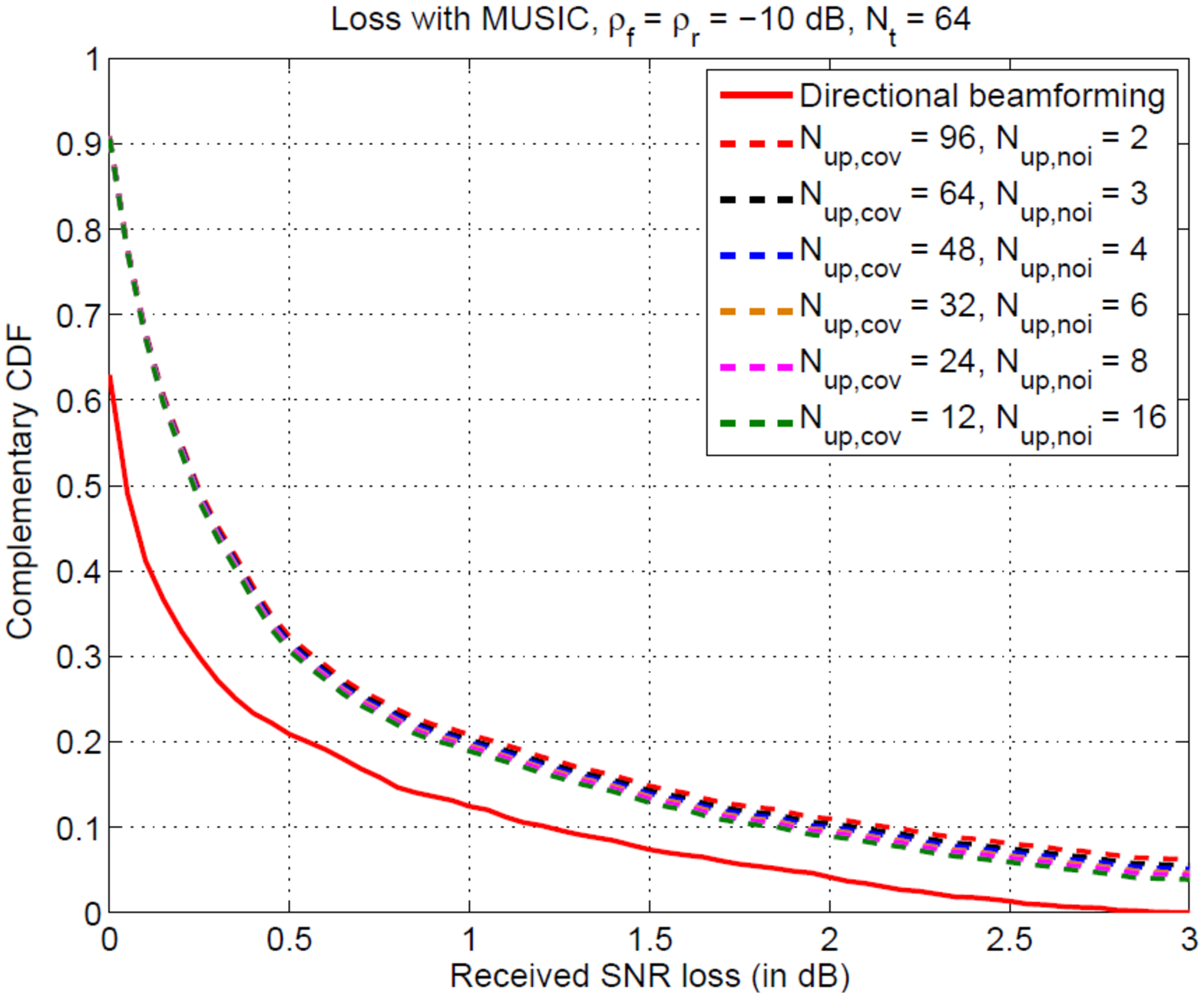}
\\ {\hspace{0.35in}} (a) & {\hspace{0.12in}} (b)
\end{tabular}
\caption{\label{fig4}
Complementary CDF of loss in ${\sf SNR}_{\sf rx}$ between the optimal
beamforming scheme in ${\cal F}_2^{N_t}$ in~(\ref{eq_2}) and MUSIC with
(a) $\rho_{\sf f} = \rho_{\sf r} = -25$ dB and (b) $\rho_{\sf f} =
\rho_{\sf r} = -10$ dB.}
\end{center}
\vspace{-5mm}
\end{figure*}

From Fig.~\ref{fig4}, we note that the performance is rather poor at low link margins, but
significantly better as the link margin improves. An important reason for the poor performance
of the MUSIC approach is that consistent covariance matrix estimation becomes a difficult
exercise with very few samples, especially as the antenna dimensions increase at the MWB end.
Furthermore, as with the noisy power iteration scheme, MUSIC also requires a non-broadcast
system design. It also suffers from a high computational complexity (dominated by the
eigen-decomposition of an $N_t \times N_t$ matrix in uplink training). In general, the
computational complexity of MUSIC can be traded off by constraining the antenna array
structure in various ways. Nevertheless, we expect the computational complexity of other
such constrained AoA/AoD learning techniques such as Estimation of Signal Parameters via
Rotational Invariance Techniques (ESPRIT) algorithm~\cite{roy_kailath}, Space-Alternating
Generalized Expectation maximization (SAGE) algorithm~\cite{fessler,fleury}, higher-order
singular value decomposition, RIMAX~\cite{richter}, and compressive sensing techniques that
employ nuclear norm optimization~\cite{ramasamy,schniter,recht} to be of similar nature as the MUSIC algorithm.
All these reasons suggest that while the MUSIC algorithm may be useful for beam refinement
after the UE has been discovered, its utility in UE discovery is questionable.

\begin{figure*}[htb!]
\begin{center}
\includegraphics[height=3.0in,width=4.5in]{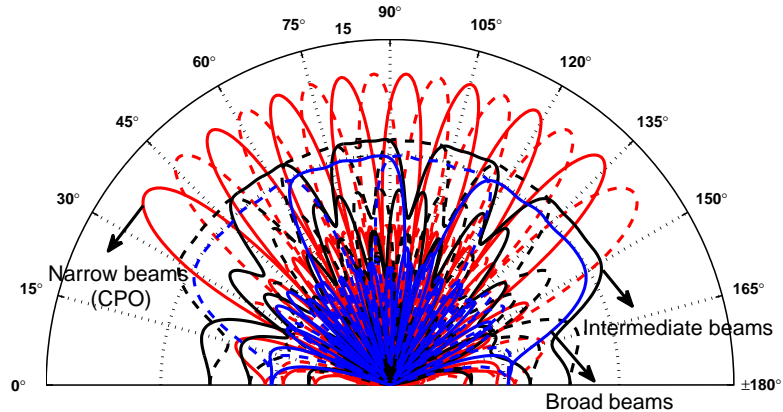}
\caption{\label{fig4c}
Main idea of beam broadening illustrated with the overlaid beam patterns of beamforming
vectors from three (narrow/${\sf CPO}$ beams, intermediate beams and broad beams) codebooks.}
\end{center}
\vspace{-5mm}
\end{figure*}

\subsection{Beam Broadening for Initial UE Discovery}
\label{sec_4b}
Let $\Omega$ denote the beamspace transformation at the MWB side, $\Omega = k d_{\sf T}
\cos(\phi_{\sf T}) = \pi \cos( \phi_{\sf T})$, corresponding to an inter-antenna spacing of
$\lambda/2$. Towards the goal of alternate direction learning strategies, we consider
the problem of understanding the tradeoff in the design of beamforming vectors that cover a
beamspace area of $\overline{\Omega}$ with as few beamforming vectors as possible without
sacrificing the worst-case beamforming gain in the coverage area~\cite{hur,rajagopal}.

The basic idea of beam broadening is illustrated in Fig.~\ref{fig4c} where three
different codebooks of beamforming vectors are used to cover a coverage area of
$120^{\sf o}$ (from $30^{\sf o}$ to $150^{\sf o}$). The first codebook (illustrated
in red) consists of narrow ${\sf CPO}$ beams (pointing at optimally chosen directions
over the coverage area) which leads to a peak beamforming gain of $10 \log_{10}(N_t)$
dB as well as a reasonably high worst-case beamforming gain over the coverage area,
although at the cost of a high UE discovery latency corresponding to a beam sweep over
$16$ directions/beams. The second codebook (illustrated in black) consists of intermediate
beams which leads to a lower peak beamforming gain (as well as a worst-case gain) over
the coverage area, but the UE discovery period is shortened as it now consists of a beam
sweep over $8$ directions/beams. The third codebook (illustrated in blue) consists of
broad beams which leads to a more reduced peak beamforming gain, but the UE discovery
period is a sweep over only $4$ directions/beams. Either codebook could be useful
for initial UE discovery depending on the link margin of the UE's involved. For example,
a UE geographically close to the MWB and suffering minimal path loss can accommodate a
broad beam codebook and be quickly discovered, whereas a UE at the cell-edge or suffering
from huge blocking losses may need the narrow ${\sf CPO}$ beam codebook to even close the
link with the MWB. The intermediate codebook trades off these two properties.

We now recast the above idea in the form of a well-posed optimization problem. For this,
given a beamspace coverage area of $\Omega_0$ for a single beam (centered around
$\Omega_{\sf c} = \pi/2$, without loss in generality), 
we seek the design of:
\begin{eqnarray}
{\bf f}_{ \Omega_0} \triangleq \arg \max \limits_{ {\bf f} \hsppp \in
\hsppp {\cal F}_{\infty}^{N_t} } \min \limits_{ \Omega \hsppp \in \hsppp
\left[ \Omega_{\sf c} - \frac{ \Omega_0}{2} , \hsppp \Omega_{\sf c} +
\frac{ \Omega_0}{2} \right] } | {\bf F}(\Omega)|^2
\label{eq31}
\end{eqnarray}
where ${\bf F}(\Omega) = \sum_{n = 0}^{N_t - 1} {\bf f}(n) e^{-j \Omega n}
= {\bf a}(\Omega)^H {\bf f}$ with ${\bf a}(\Omega) = \left[ 1, \hsppp e^{j \Omega},
\cdots, e^{ j(N_t-1) \Omega} \right]^T$.
With ${\bf f}_{\Omega_0}$ as template, the number of beamforming vectors needed to
cover $\overline{\Omega}$ (say, a $120^{\sf o}$ field-of-view as in Fig.~\ref{fig4c})
is ${\sf No. \hsppp beams} = \frac{ \overline{\Omega}} { \Omega_0}$.

We start with an upper bound on the tradeoff between ${\sf No. \hsppp beams}$ and
the worst-case beamforming gain over $\overline{\Omega}$. From the Parseval identity,
we have the following trivial relationship for any ${\bf f}$:
\begin{eqnarray}
\frac{1}{2\pi} \int_{-\pi}^{\pi} | {\bf F}(\Omega)|^2 d \Omega =
\sum_{n = 0}^{N_t - 1} |{\bf f}(n)|^2 \leq 1.
\label{eq_19}
\end{eqnarray}
If $\min \limits_{ \Omega \hsppp \in \hsppp \left[ \Omega_{\sf c} - \frac{ \Omega_0}{2} ,
\hsppp \Omega_{\sf c} + \frac{ \Omega_0}{2} \right] } | {\bf F}(\Omega)|^2 = P$, we have
\begin{eqnarray}
\frac{1}{2\pi} \int_{-\pi}^{\pi} | {\bf F}(\Omega)|^2 d \Omega
\geq \frac{1}{2\pi} \int_{ \Omega \hsppp \in \hsppp \left[ \Omega_{\sf c} - \frac{ \Omega_0}{2} ,
\hsppp \Omega_{\sf c} + \frac{ \Omega_0}{2} \right] } | {\bf F}(\Omega)|^2 d \Omega
\geq P \cdot \frac{ \Omega_0}{ 2 \pi}
\Longrightarrow P \leq \frac{ 2 \pi}{ \Omega_0}.
\label{eq33}
\end{eqnarray}
Further, $P$ is also constrained as $P \leq N_t$ since the maximal beamforming gain cannot
exceed $N_t$ in any direction. Thus, the worst-case beamforming gain over
this area (in dB) is upper bounded as
\begin{eqnarray}
{\sf BF \hspp Gain} & \triangleq & 
10 \log_{10}\left( P \right) \leq 10 \log_{10} \left( \min
\left(N_t, \hsppp \frac{ 2 \pi} { \overline{\Omega} } \cdot {\sf No. \hsppp beams}
\right) \right).
\label{eq35}
\end{eqnarray}
We now provide an alternate non-trivial approach based on computation of eigenvalues of
certain appropriately-defined matrices for a better upper bound of this tradeoff.

\begin{figure*}[htb!]
\begin{center}
\begin{tabular}{cc}
\includegraphics[height=2.7in,width=3.15in] {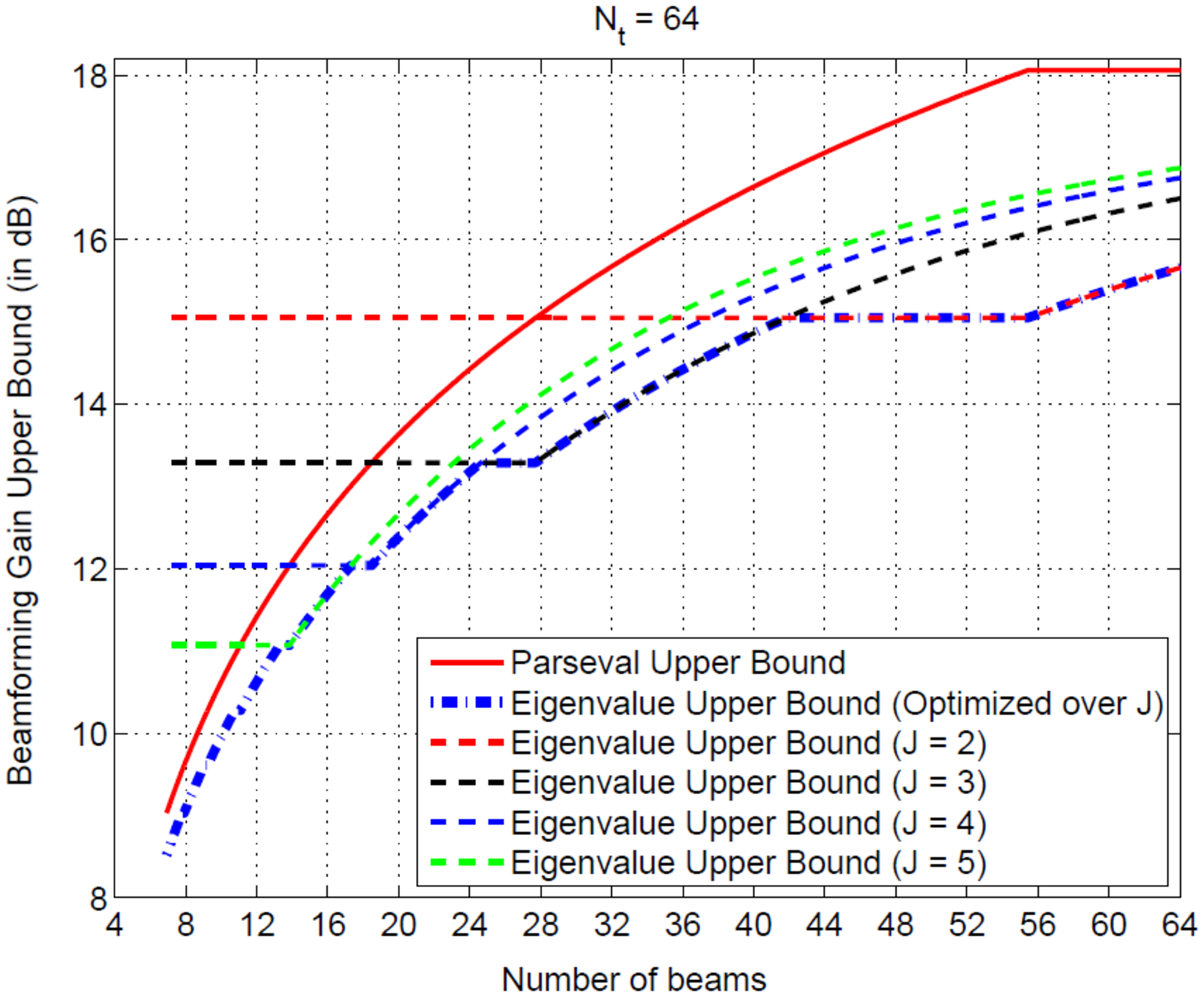}
&
\includegraphics[height=2.7in,width=3.15in] {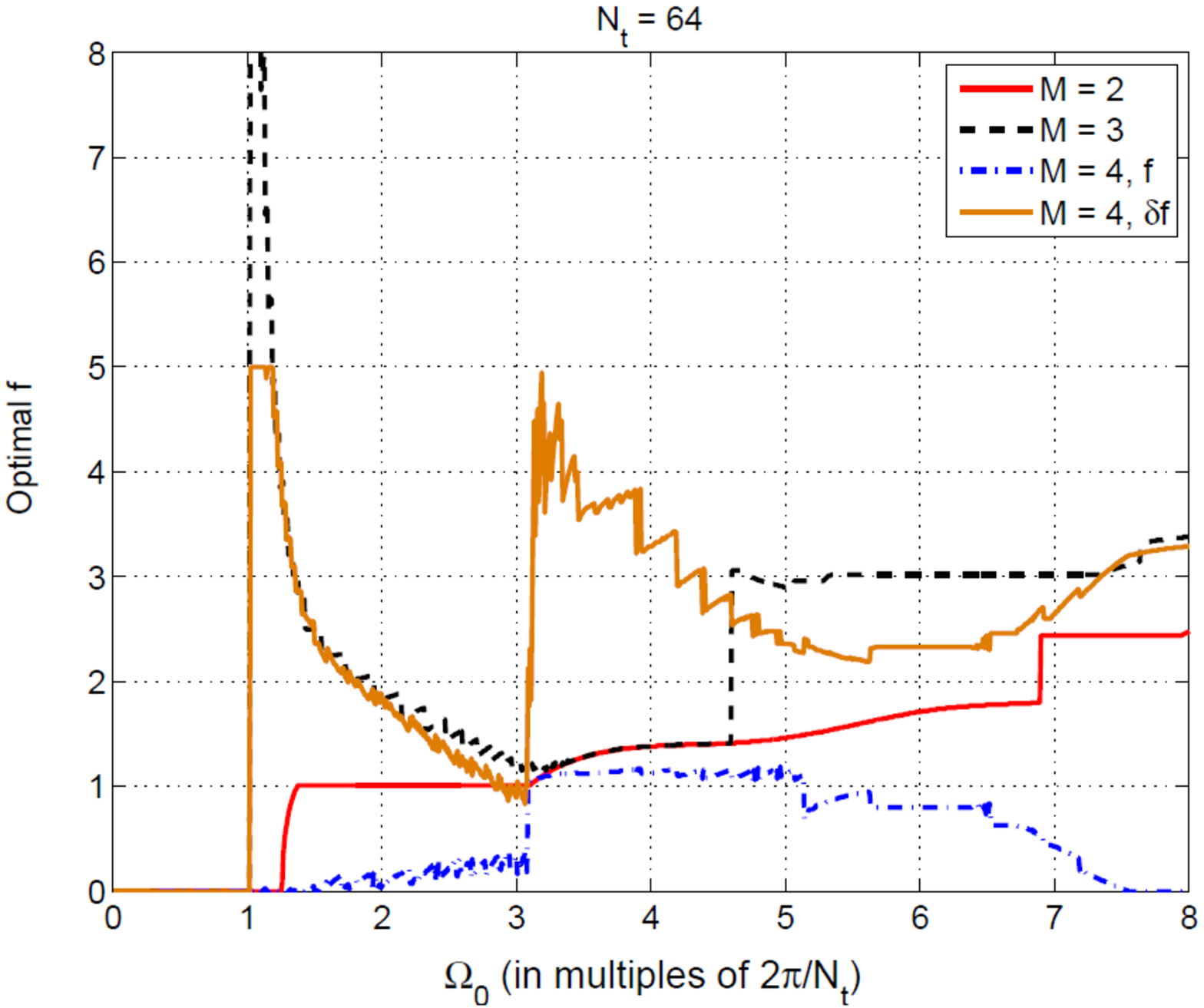}
\\ {\hspace{0.35in}} (a) & {\hspace{0.12in}} (b)
\\
\includegraphics[height=2.7in,width=3.15in] {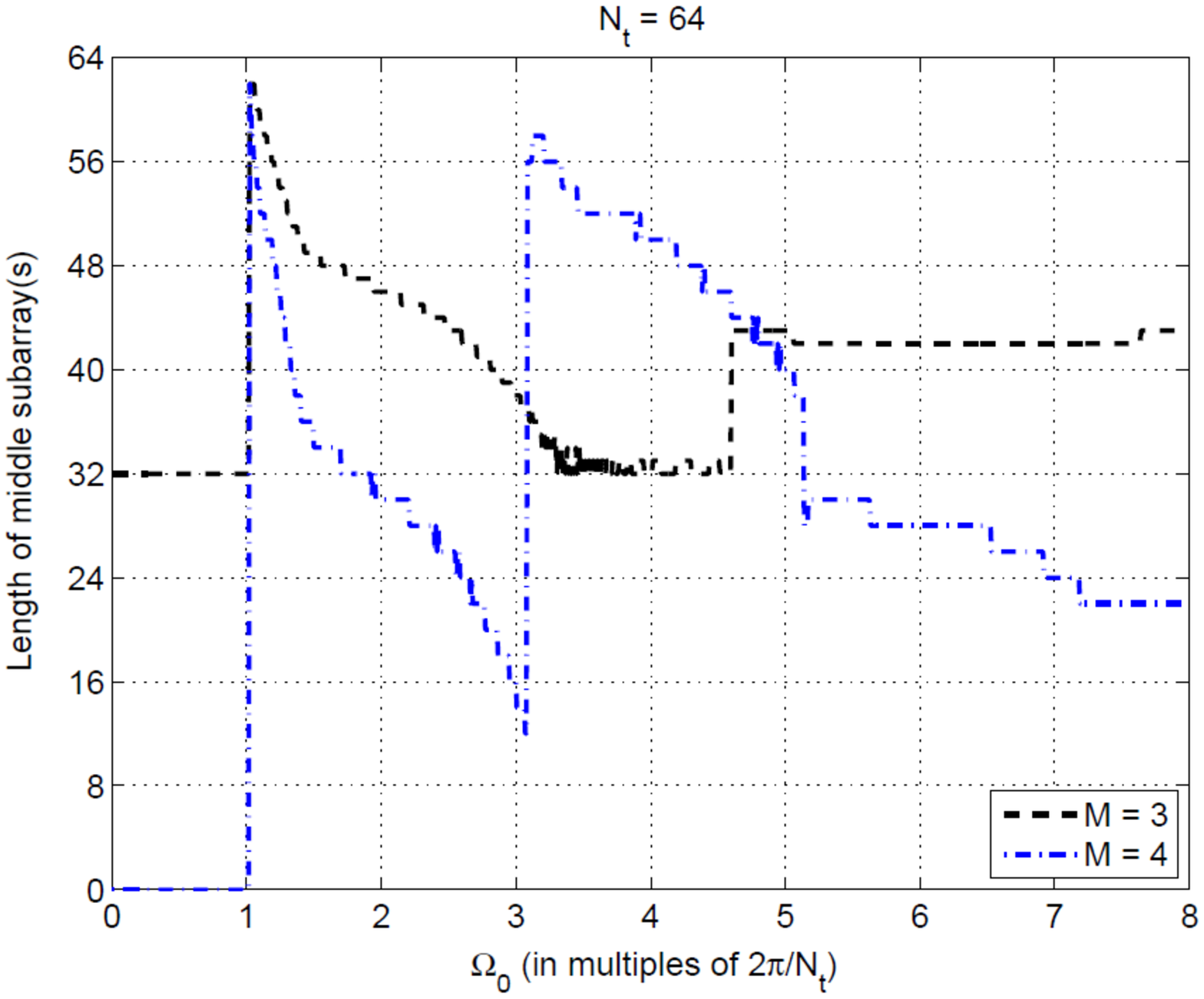}
&
\includegraphics[height=2.7in,width=3.15in] {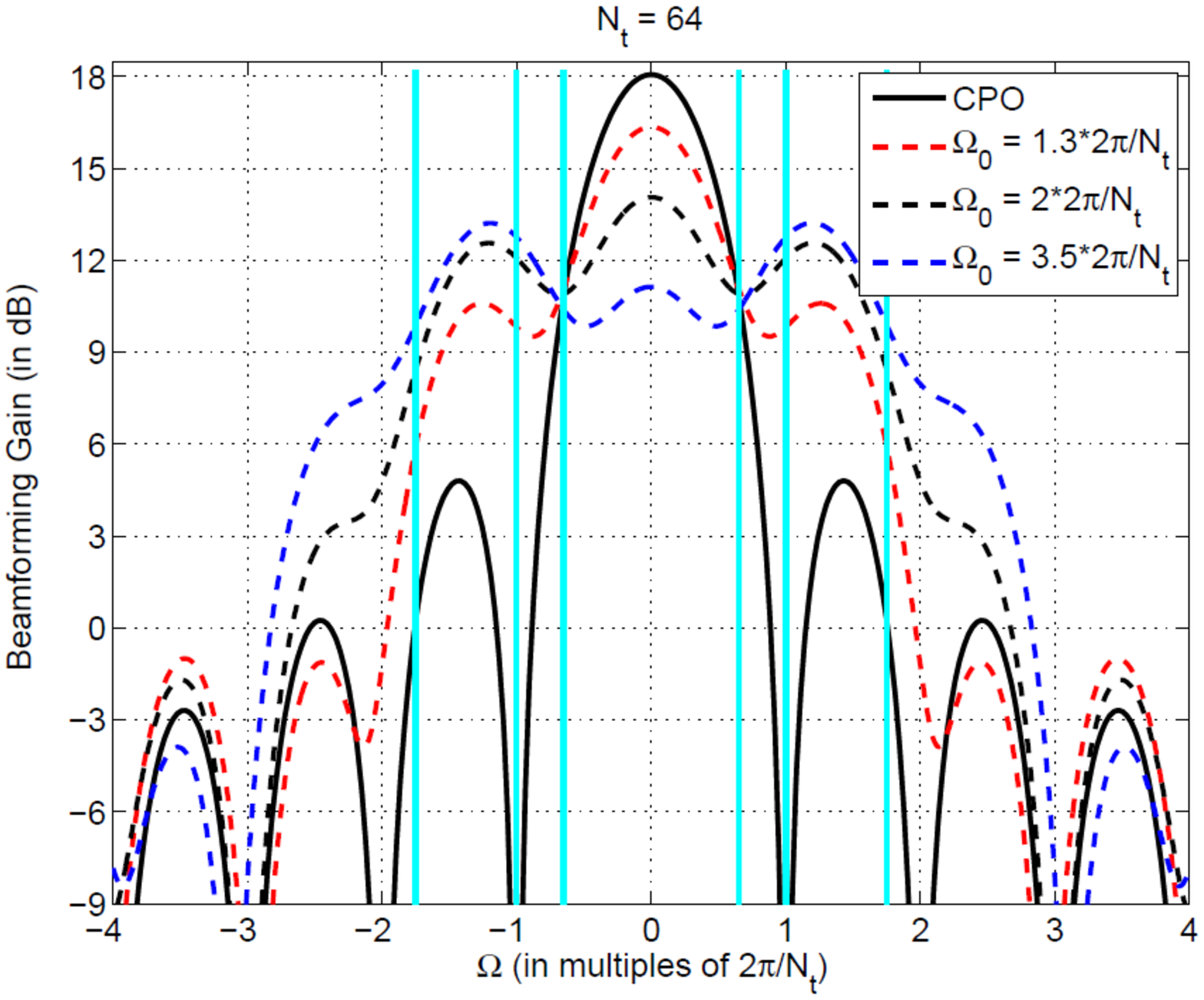}
\\ {\hspace{0.35in}} (c) & {\hspace{0.12in}} (d)
\end{tabular}
\caption{\label{fig5}
(a) Upper bound to the beamforming gain vs.\ number of beams tradeoff for
$N_t = 64$. (b) Optimal ${\sf f}$ and (c) length of middle subarray(s)
as a function of $\Omega_0$ for different choices of $M$. (d) Beam patterns
of broadened beams for certain choices of $\Omega_0$ in the $M = 2$ case.
}
\end{center}
\vspace{-5mm}
\end{figure*}

\begin{thm}
\label{thm3}
Let $\Omega_j, \hspp j = 1, \cdots , J$ be a set of sampling frequencies over the
beamspace area of $\Omega_0$ spanned by the beamforming vector ${\bf f}$. The worst-case
beamforming gain with ${\bf f}$ is upper bounded by the solution to the following
optimization:
\begin{eqnarray}
{\sf BF \hspp Gain} \leq 10 \log_{10} \left(
N_t, \hspp \min \limits_{ J, \hsppp \{ \Omega_j \} }
\frac{ \lambda_{\sf max} \left(
\sum_{j = 1}^J {\bf a}(\Omega_j) {\bf a}(\Omega_j)^H \right) }{J}
\right)
\label{eq_23}
\end{eqnarray}
where $\Omega_j \in \left[ \Omega_{\sf c} - \frac{ \Omega_0 }{ 2} , \hsppp
\Omega_{\sf c} + \frac{ \Omega_0 }{2} \right]$ with $\Omega_0 = \frac{ \overline{\Omega}}
{ {\sf No. \hsppp beams} }$ and $\lambda_{\sf max}(\bullet)$ stands for the
largest eigenvalue of the underlying positive semi-definite matrix.
\end{thm}
\begin{proof}
See Appendix~\ref{appD}.
\end{proof}

Fig.~\ref{fig5}(a) numerically optimizes the expression in~(\ref{eq_23}) and
plots the upper bound to the worst-case beamforming gain as a function of the
number of beams to cover a $120^{\sf o}$ field-of-view with $N_t = 64$. For
the eigenvalue-based approach, we plot the upper bound for specific choices of
$J$ with $\{ \Omega_j \}$ optimized, as well as the upper bound based on a joint
optimization over $J$ and $\{ \Omega_j \}$. Note that the horizontal segments in
the joint optimization correspond to the fact that the tradeoff with a larger
number of beams can be no worser than the tradeoff with a smaller number of beams.

In contrast to the upper bound, we now propose specific approaches towards the
goal of beam broadening. For this, we initially consider partitioning of the antenna
array at the MWB side into {\em virtual subarrays} where each virtual subarray is
used to beamform to a certain appropriately-chosen virtual direction. The expectation
from this approach is that the beam patterns from the individual virtual subarrays
combine to enhance the coverage area of the resultant beam with minimal loss in peak
gain due to reduction in the effective aperture of the subarrays.

\begin{figure*}[htb!]
\begin{center}
\begin{tabular}{cc}
\includegraphics[height=2.7in,width=3.15in] {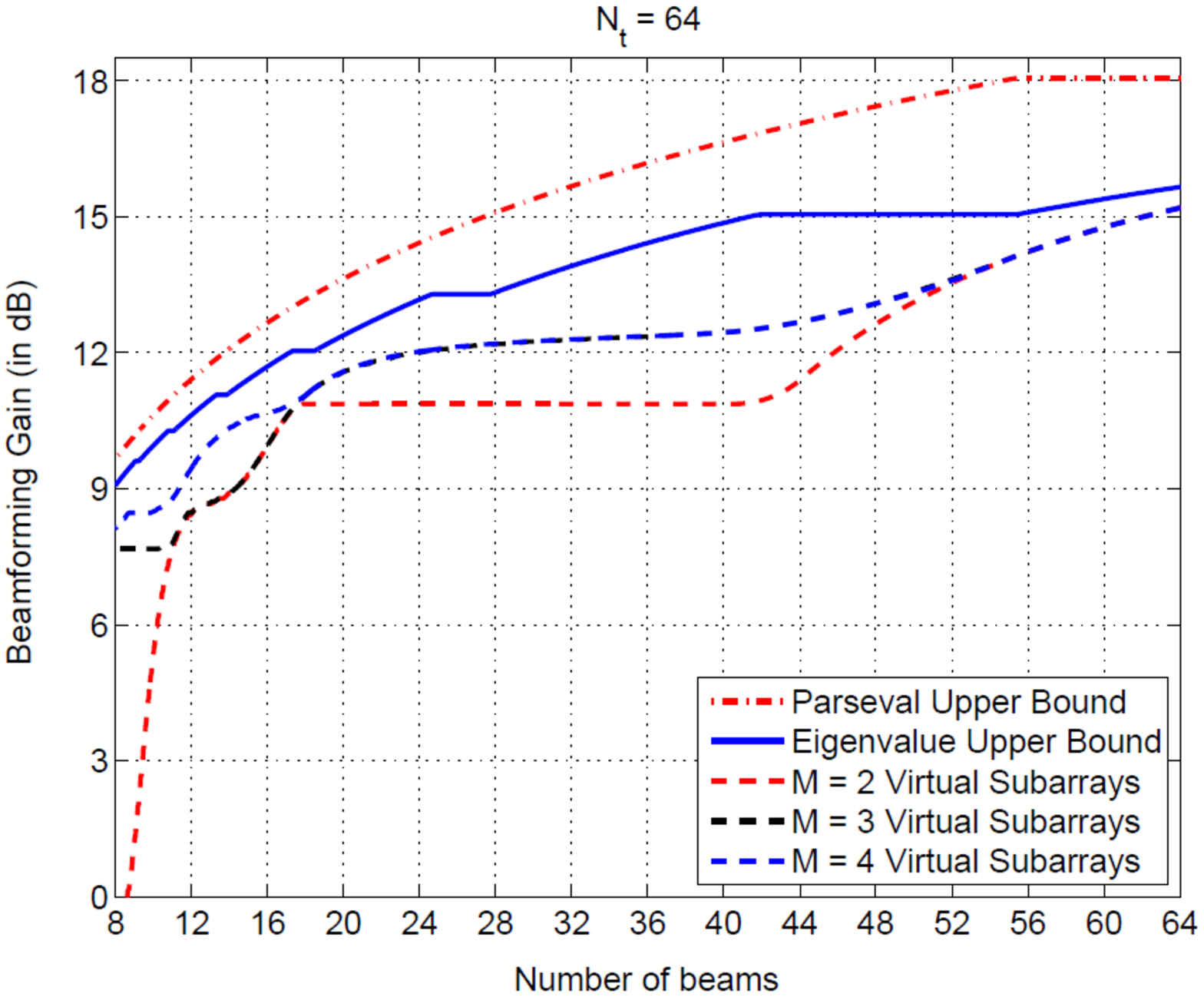}
&
\includegraphics[height=2.7in,width=3.15in] {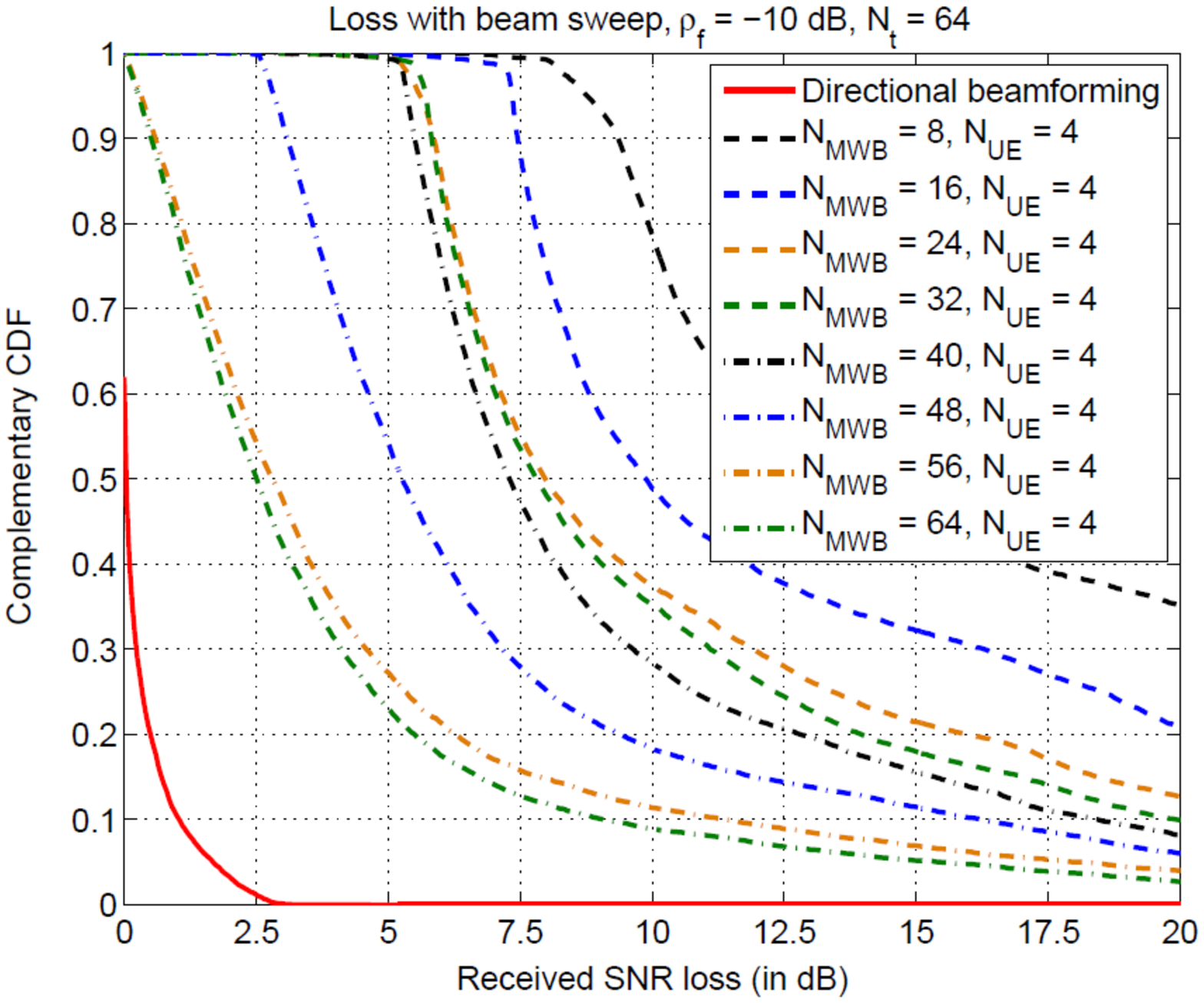}
\\ {\hspace{0.35in}} (a) & {\hspace{0.12in}} (b)
\end{tabular}
\caption{\label{fig6}
(a) Performance tradeoff with different beam broadening approaches for $N_t = 64$.
(b) Complementary CDF of loss in ${\sf SNR}_{\sf rx}$ between the optimal
beamforming scheme in ${\cal F}_2^{N_t}$ in~(\ref{eq_2}) and beam sweep with the
broadened beams as in Fig.~\ref{fig6}(a) at $\rho_{\sf f} = -10$ dB.}
\end{center}
\vspace{-5mm}
\end{figure*}

As a specific example, in the $M = 2$ virtual subarray setting, we propose the
following beamforming vector that orients along
$\pi - \cos^{-1} \left( \frac{2 {\sf f}} {N_t} \right)$ and
$\cos^{-1} \left (  \frac{2 {\sf f} } {N_t} \right)$ with each half of the array.
That is,
\begin{eqnarray}
{\bf f}(n) = \frac{1}{ \sqrt{N_t} } \left\{
\begin{array}{ll}
\exp \left( - \frac{j 2 \pi {\sf f}}{N_t} ( n - \frac{N_t}{2} + \frac{1}{2} ) \right)
& {\sf if} \hspp 0 \leq n \leq \frac{N_t}{2} - 1 \\
\exp \left( \frac{j 2 \pi {\sf f}}{N_t} ( n - \frac{N_t}{2} + \frac{1}{2} ) \right)
& {\sf if} \hspp \frac{N_t}{2} \leq n \leq N_t - 1
\end{array}
\right.
\label{eq_24}
\end{eqnarray}
where ${\bf f}$ is designed to be a broadened beam by optimally choosing ${\sf f}$.
Similarly, in the $M = 3$ and $M = 4$ settings, we propose the following beamforming
vectors with two parameters (${\sf f}$ and $0 \leq L \leq \frac{N_t}{2}$)
and three parameters (${\sf f}$, \hsppp $\delta {\sf f}$ and $0 \leq L \leq \frac{N_t}{2}$),
respectively:
\begin{eqnarray}
{\bf f} & = &
\frac{1}{ \sqrt{N_t} } \left\{
\begin{array}{ll}
\exp \left(- \frac{j 2 \pi {\sf f}}{N_t} ( n - \frac{N_t}{2} + \frac{1}{2} + L) \right)
& {\sf if} \hspp 0 \leq n \leq \frac{N_t}{2}  - L - 1 \\
1 & {\sf if} \hspp \frac{N_t}{2} - L \leq n \leq \frac{N_t}{2} + L - 1 \\
\exp \left( \frac{j 2 \pi {\sf f}}{N_t} ( n - \frac{N_t}{2} + \frac{1}{2} - L) \right)
& {\sf if} \hspp \frac{N_t}{2} + L \leq n \leq N_t - 1,
\end{array}
\right.
\label{eq_25}
\\
{\bf f} & = &
\frac{1}{ \sqrt{N_t} } \left\{
\begin{array}{ll}
\exp \left(- \frac{j 2 \pi ( {\sf f} + \delta {\sf f})  }{N_t}
( n - \frac{N_t}{2} + \frac{1}{2} ) - \frac{j 2 \pi \cdot \delta{\sf f}}
{N_t} \left( L - \frac{1}{2} \right) \right)
& {\sf if} \hspp 0 \leq n \leq \frac{N_t}{2} - L - 1 \\
\exp \left( - \frac{j 2 \pi {\sf f}} {N_t} (n - \frac{N_t}{2} + \frac{1}{2} ) \right)
& {\sf if} \hspp \frac{N_t}{2} - L \leq n \leq \frac{N_t}{2} - 1 \\
\exp \left(  \frac{j 2 \pi {\sf f}} {N_t} (n - \frac{N_t}{2} + \frac{1}{2} ) \right)
& {\sf if} \hspp \frac{N_t}{2} \leq n \leq \frac{N_t}{2} + L - 1 \\
\exp \left( \frac{j 2 \pi ( {\sf f} + \delta {\sf f})  }{N_t}
( n - \frac{N_t}{2} + \frac{1}{2} ) - \frac{j 2 \pi \cdot \delta{\sf f}}
{N_t} \left( L - \frac{1}{2} \right) \right)
& {\sf if} \hspp \frac{N_t}{2} + L \leq n \leq N_t - 1.
\end{array}
\right.
\label{eq_26}
\end{eqnarray}
Note that~(\ref{eq_25}) reduces to the $M = 2$ setting in~(\ref{eq_24}) with $L = 0$
where the beams are pointed at ${\sf f}$, and with $L = \frac{N_t}{2}$ where the beams
are pointed at ${\sf f} = 0$. Similarly,~(\ref{eq_26}) reduces to the $M = 2$ setting
with $L = 0$ and $L = \frac{N_t}{2}$ where the beams are pointed at ${\sf f} + \delta {\sf f}$
and ${\sf f}$, respectively.

Fig.~\ref{fig5}(b) plots the optimal values of ${\sf f}$ (and $\delta {\sf f}$)
designed to maximize $\min |{\bf F}(\Omega)|^2$ as a function of $\Omega_0$ with
$M = 2, 3, 4$ and $N_t = 64$. Fig.~\ref{fig5}(c) plots the length of the middle
subarray in the $M = 3$ case (middle subarrays in the $M = 4$ case) as a function of
$\Omega_0$. From these two plots, we see that for small values of $\Omega_0$,
choosing ${\sf f} = 0$ is optimal, whereas the length of the middle subarray
decreases as $\Omega_0$ increases corresponding to gradual beam orientation away
from ${\sf f} = 0$. Fig.~\ref{fig5}(d) plots the shape of the broadened beams so
optimized for three choices of $\Omega_0$: $\frac{ \Omega_0}{ 2 \pi/N_t} =
\{ 1.3, \hspp 2, \hspp 3.5 \}$. Also, plotted are vertical lines at
$\frac{ \Omega_0}{ 2 \pi/N_t}  = \pm 1.3, \hspp \pm 2, \hspp \pm 3.5$. From this
plot, we see that within their corresponding regimes, each broadened beam maximizes
$\min |{\bf F}(\Omega)|^2$. In addition, Fig.~\ref{fig6}(a) captures the tradeoff
between the number of beamforming vectors and the worst-case beamforming gain with
the $M = 2, 3, 4$ subarray scheme in the $N_t = 64$ case. Clearly, across all
regimes of interest of ${\sf No. \hspp beams}$, the $M = 4$ subarray scheme is
within a couple of dB of the upper bound in terms of beamforming gain illustrating
the utility of the proposed approach.

\subsection{Learning Dominant Directions via Beam Sweep with Broadened Beam Codebooks}
\label{sec_4c}

We use the template broadened beamforming vectors designed in Sec.~\ref{sec_4b}
corresponding to different beam broadening factors (and their shifted versions) to
design a beam sweep codebook for the MWB side. In particular, we consider those beam
broadening factors that lead to $N_{\sf mwb} = \{8, \hspp 16, \hspp 24, \hspp 32, \hspp
40, \hspp 48, \hspp 56,  \hspp 64 \}$ elements in the beam sweep codebook at the MWB.
On the other hand, since $N_r = 4$ at the UE side, a simpler codebook of $N_{\sf ue}
= 4$ beamforming vectors corresponding to an equal partition of the field-of-view is
sufficient for the purpose of beam sweep. With these different codebook choices at
the MWB and UE, we find the best choice of beamforming vectors that maximize
${\sf SNR}_{\sf rx}$ and use them for subsequent beamforming/beam refinement. An
important advantage of the beam sweep approach is that it allows a broadcast solution
(the same codebook can be reused across multiple UEs within the field-of-view).

\begin{figure*}[htb!]
\begin{center}
\begin{tabular}{cc}
\includegraphics[height=2.7in,width=3.15in] {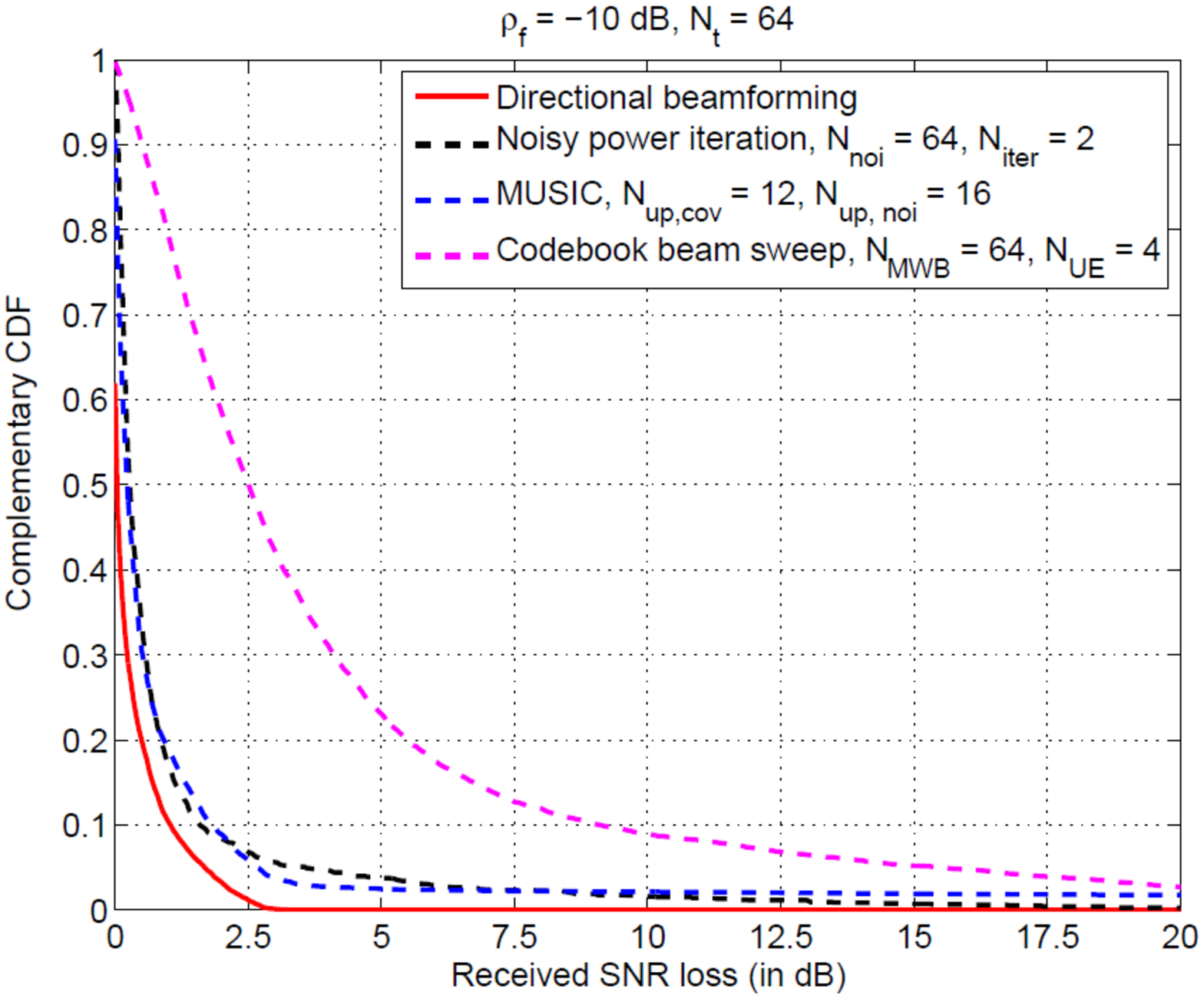}
&
\includegraphics[height=2.7in,width=3.15in] {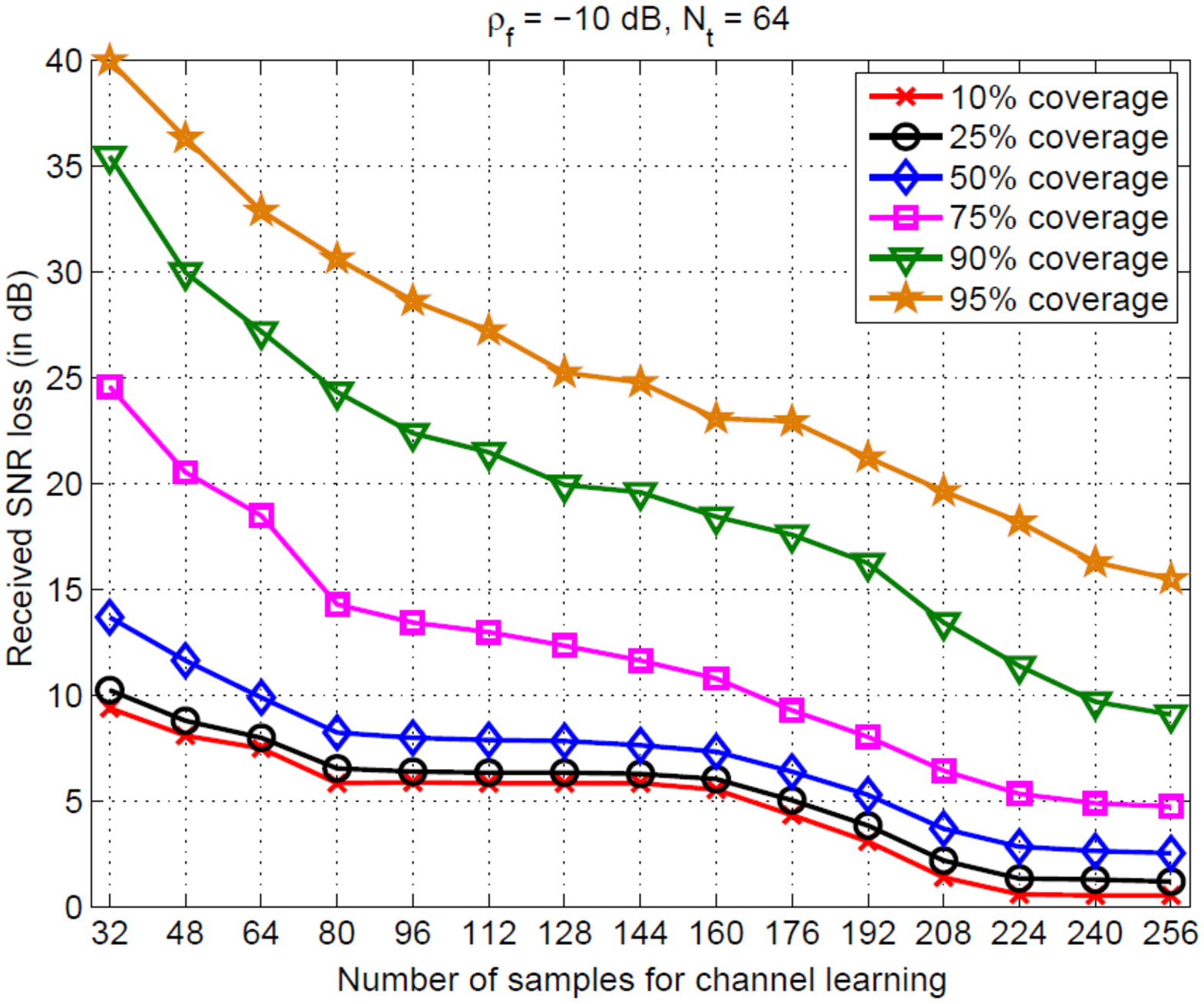}
\\ {\hspace{0.35in}} (a) & {\hspace{0.12in}} (b)
\end{tabular}
\caption{\label{fig7}
(a) Comparison of complementary CDF of ${\sf SNR}_{\sf rx}$ loss with different
beamforming approaches at $\rho_{\sf f} = \rho_{\sf r} = -10$ dB. (b) Tradeoff
between number of samples for channel learning with the beam sweep approach and loss
in ${\sf SNR}_{\sf rx}$ in initial UE discovery. }
\end{center}
\vspace{-5mm}
\end{figure*}

Fig.~\ref{fig6}(b) plots the complementary CDF of the loss in ${\sf SNR}_{\sf rx}$
with the beam sweep approach at $\rho_{\sf f} = -10$ dB. Fig.~\ref{fig7}(a) compares
the complementary CDF of loss in ${\sf SNR}_{\sf rx}$ with the best noisy power
iteration scheme, MUSIC algorithm and the beam sweep approach at the same $\rho_{\sf f}$
value. Clearly, the beam sweep approach has a poorer performance relative to the
other schemes, but its simplicity results in a better system design than possible with
the other approaches. Further, Fig.~\ref{fig7}(b) plots the loss in ${\sf SNR}_{\sf rx}$
at different percentile levels ($10\%, \hsppp 25\%, \hsppp 50\%, \hsppp 75\%, \hsppp
90\%$ and $95\%$) as a function of the number of samples used in channel learning with
the beam sweep scheme. From this study, we note that at small coverage levels, certain
codebook size choices are better in the tradeoff curve than other choices (for example,
a $N_{\sf mwb} = 20$ codebook over a $N_{\sf mwb} = 40$ codebook) and these advantages
correspond to the steepness of the achievability curve (see Fig.~\ref{fig6}(a)).

Further, while ${\sf SNR}_{\sf rx}$ improves as the codebook size increases, good
users with a better link margin (e.g., users with a smaller path loss) can be discovered
with a lower discovery latency (corresponding to a smaller codebook size) than those
cell-edge/blocked users with a worser link margin. Fig.~\ref{fig7}(b) suggests a
smooth roll-off in the discovery latency of the users with a worser link margin. That
said, the beam sweep approach could indeed suffer a significant loss in performance
especially with a cell-edge/blocked user. In such scenarios, the design for such a
user could include coding over long sequences for enhanced time-repetition/processing
gain, high MWB densification, a low-frequency overlay of multiple narrow CPO beams with
a worst-case beamforming gain (as close to the $10 \log_{10}(N_t)$ dB peak gain) in the
coverage area, among many approaches. However, such a design could lead to a significant
drag on the performance tradeoffs of the good/median user. Thus, they could be initiated
by the UE when it perceives a poor link on a unicast basis.

\begin{figure*}[htb!]
\begin{center}
\begin{tabular}{cc}
\includegraphics[height=2.7in,width=3.15in] {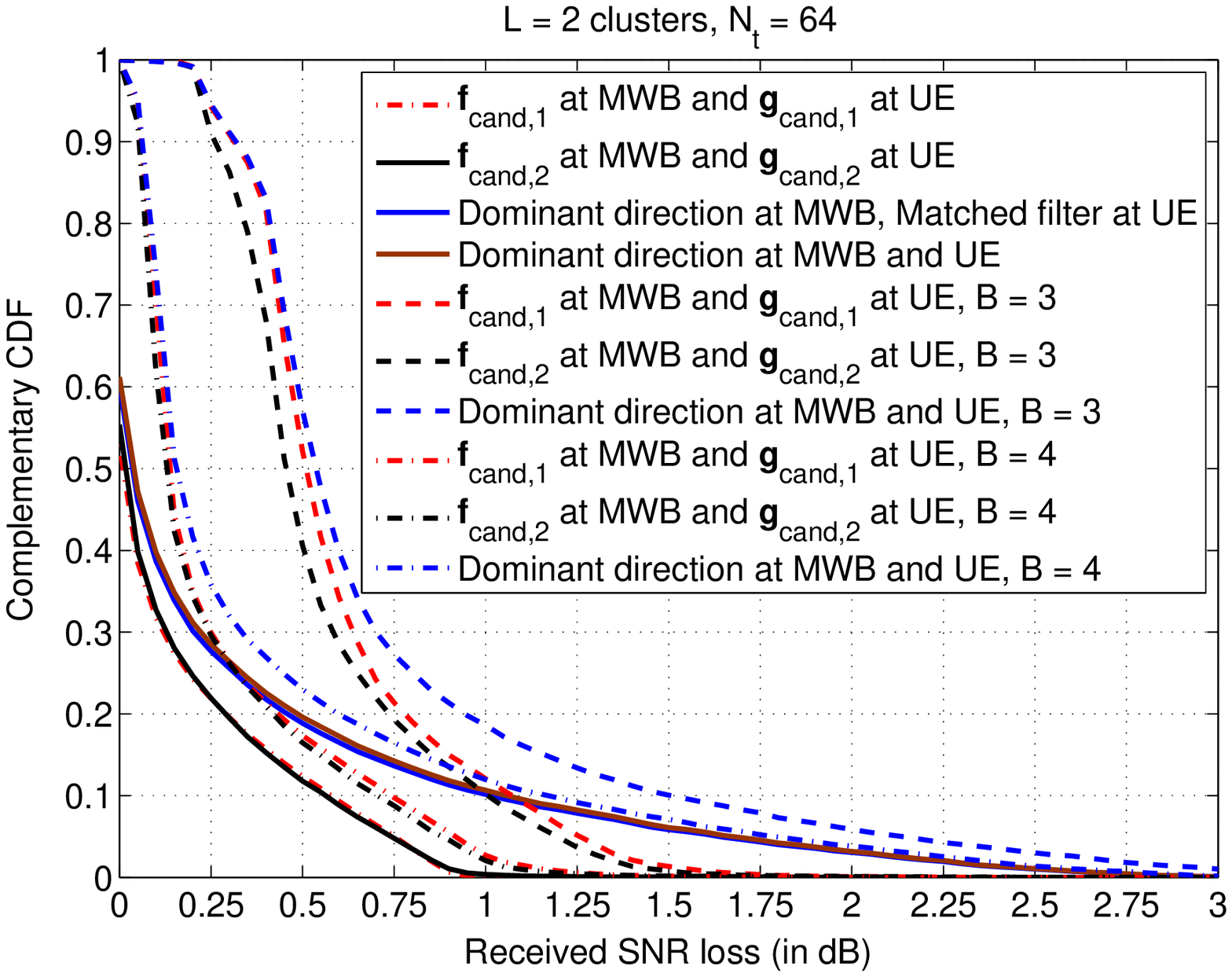}
&
\includegraphics[height=2.7in,width=3.15in] {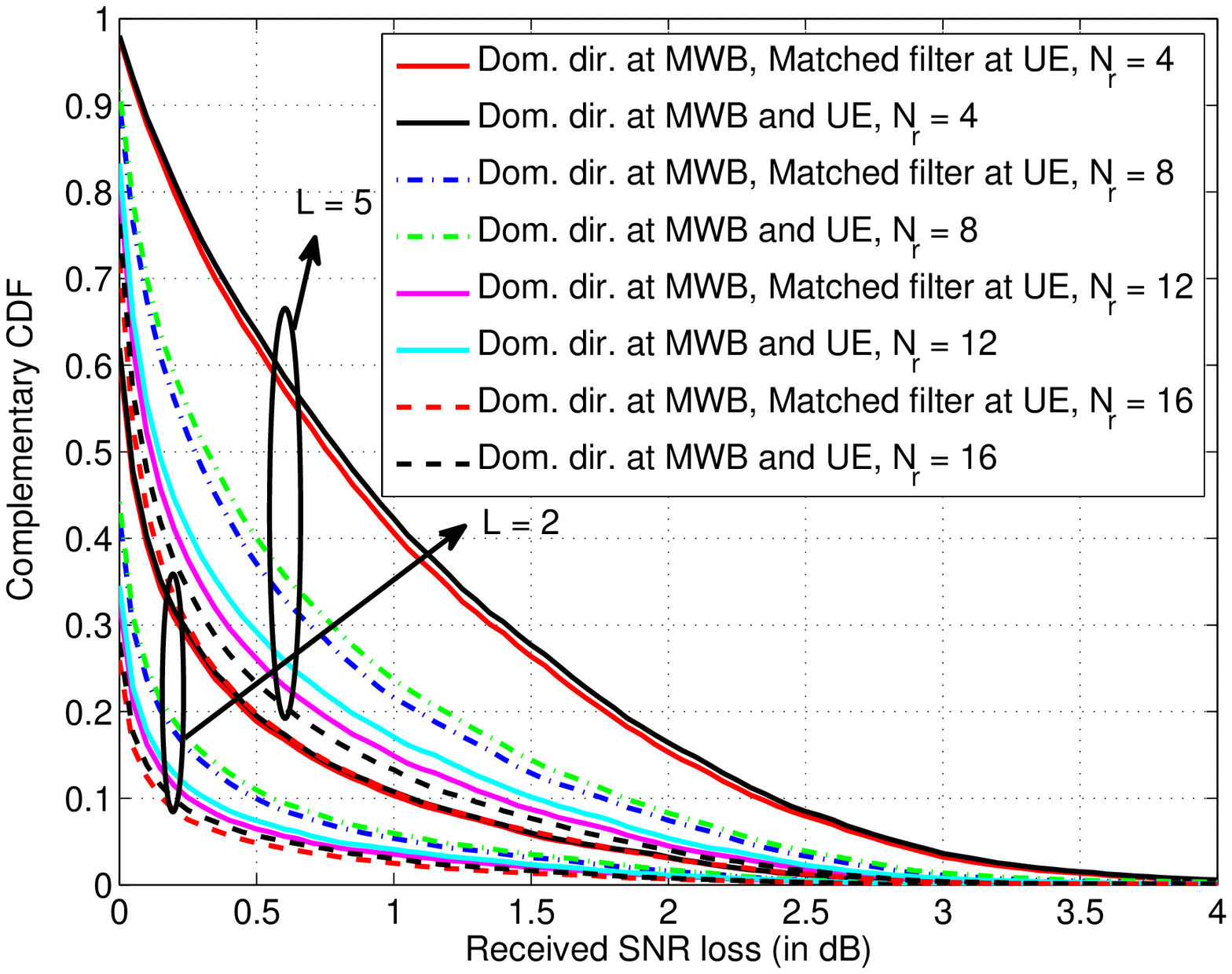}
\\ {\hspace{0.35in}} (a) & {\hspace{0.12in}} (b)
\end{tabular}
\caption{\label{fig8}
(a) Loss in ${\sf SNR}_{\sf rx}$ with different RSV-type beamforming schemes with
perfect phase of beamformers, and finite $B = 3$ and $B = 4$ bit phase quantization.
(b) Complementary CDF of loss in ${\sf SNR}_{\sf rx}$ as a function of $N_r$ with
$L = 2$ and $L = 5$ clusters.
}
\end{center}
\vspace{-5mm}
\end{figure*}

\section{Comments on Practical Applications}

\subsection{Finite-Bit Phase Shifters}
The entire focus of this work has been on beamformers that can be realized with
infinite-precision phase shifters allowing an arbitrary phase resolution for the
beamformer weights. However, in practice, beamformers at both the MWB and UE ends
are constrained to use finite-bit phase shifters. Nevertheless, our studies suggest
that even a $B = 3$ or $B = 4$ bit phase shifter (which is practically realizable
at low cost) is sufficient. Specifically, Fig.~\ref{fig8}(a) considers the $L = 2$,
$N_r = 4$ and $N_t = 64$ case considered in Fig.~\ref{fig1}(a) with $B = 3$ and
$B = 4$ bit phase shifters and plots the loss in ${\sf SNR}_{\sf rx}$ relative to
the infinite-precision optimal beamforming scheme. From Fig.~\ref{fig8}(a), we observe
a minimal loss (less than $0.25$ dB) with a $B = 4$ bit phase shifter, thereby
suggesting that quantization is not a serious detriment to the performance of mmW
MIMO systems. This justifies focus on infinite-precision beamformers in this work.

\subsection{More Antennas at the UE End}
A larger number of antenna elements (at the MWB and UE ends) could render mmW
systems more attractive in terms of data rates. With this backdrop, the UE-side
and MWB-side antenna numbers of $N_r = 4$ and $N_t = 64$ in our simulation studies
can be justified as follows. The MWB serves as a network resource with softer
constraints on the array aperture and hence, sectorized coverage ($90^{\sf o}$ or
$120^{\sf o}$ coverage per array at the MWB end) is likely leading to a larger number
of antennas at this end. On the other hand, multiple subarrays need to deployed at
the UE end to cover all the sectors leading to a smaller-dimensional subarray. For
example, a $16$ element antenna array consisting of $4$ subarrays of $4$ antennas
each with a $\lambda/2$ inter-antenna element spacing would still require an aperture of
$7.5$ cm at $30$ GHz --- a considerable expense in array aperture at the UE end. Thus,
the $N_r = 4$ assumption made in this work is not conservative, but quite realistic for
practical mmW systems.

As $N_r$ increases, the performance of all the directional schemes get better
relative to the optimal scheme (this is also true in terms of absolute values of
${\sf SNR}_{\sf rx}$ due to increased array gain with increasing $N_r$, but this
is not shown here) as can be seen from Fig.~\ref{fig8}(b) where the complementary
CDF of loss in ${\sf SNR}_{\sf rx}$ is plotted for the $N_t = 64$ case with
$L = 2$ and $L = 5$ clusters for four different choices of $N_r$: $N_r \in \{ 4,
\hsppp 8, \hsppp 12, \hsppp 16 \}$. While these plots correspond to perfect
directional beamforming, with reasonable $\rho_{\sf f}$, we expect the performance
of the directional learning schemes such as beam sweep or MUSIC to get similarly
better with increasing $N_r$.

\subsection{Coherence Time Constraints}
In addition to the different architectural tradeoffs in terms of system design,
the realizability of different beamforming strategies also critically depend on the
coherence time of the mmW channel. Initial measurement studies suggest that the
coherence time is on the order of a few milliseconds~\cite{adhikary,swindlehurst_mmw}.
While this coherence time is considerably (an order of magnitude) shorter than in
sub-$6$ GHz systems, this constraint only leads to a more favorable view of
directional beamforming approaches relative to RSV-type schemes due to the need for
a bi-directional feedback for the implementation of an RSV-type scheme.
Fig.~\ref{fig3}(b) also addresses the lack of robustness with an RSV-type scheme
due to small path length changes that could happen at a sub-coherence time level. On
the other hand, the essentially stable performance of the directional beamforming
scheme to such changes makes it an attractive candidate for initial UE discovery.

\subsection{Planar Antenna Arrays}
While the entire development in this paper so far assumed an ULA geometry for the
antennas, this was done primarily for the sake of illustrating the tradeoffs in the
beam broadening problem which would have been difficult with a smaller dimensional
planar array, e.g., $8 \times 8$. More general expressions can be written for the
array steering vectors when the antennas are laid out according to other geometric
configurations~\cite{oelayach,balanis}. In particular, the proposed development of
this paper can be easily extended to a planar array geometry.

\subsection{Comparison with Other Initial UE Discovery Approaches}
\label{sec_5e}
Jeong {\em et al.}~\cite{jeong_et_al} study a random access procedure for initial
UE discovery based on beam scanning and propose different approaches to enhance the
performance of beam sweep such as the use of multiple RF chains, enhanced preamble
detection, optimal cell design, etc. However, the viability of beam scanning as a
beamforming procedure for initial UE discovery (relative to other signal
processing-based techniques) is not considered in~\cite{jeong_et_al}.

Ghadikolaei {\em et al.}~\cite{ghadikolaei} also address the initial UE discovery
problem and characterize the essential tradeoffs in the design of control channels.
Initial UE search approaches such as an omni-directional beam sweep, one-sided
directional beam sweep, or a bi-directional beam sweep are considered. Since the
performance tradeoffs are addressed from a MAC layer perspective, PHY aspects
such as beamformer design under RF constraints, tradeoffs between and implications
of different beamforming approaches on mmW system design, etc. are not considered
in~\cite{ghadikolaei}. Nevertheless, our work is similar in flavor to~\cite{ghadikolaei}
in terms of the received ${\sf SNR}$ vs.\ initial UE discovery latency tradeoff that
both works quantify from PHY and MAC viewpoints, respectively.

The theme considered in Barati {\em et al.}~\cite{barati} for the initial UE discovery
problem is closely aligned with the theme of our work. The authors consider a beam sweep
procedure and derive the structure of generalized likelihood ratio detectors for detecting the
primary sync signal. Based on this study, the authors conclude that omni-directional
beam scanning provides a better tradeoff point in the received ${\sf SNR}$ vs.\ initial
UE discovery latency curve than a beam sweep. At this stage, note that an extreme case
of beam broadening considered in our work (where only one beam covers the entire coverage
area) is equivalent to an omni-directional scan considered in~\cite{barati} and thus the
issue lies in figuring out the best tradeoff point on the received ${\sf SNR}$ vs.\ initial
UE discovery latency tradeoff curve. Independent of performance comparisons, from a
system design perspective, as the authors rightly point out in~\cite{barati}, detection
of the primary sync signal via an omni-directional scan provides no knowledge of the AoD
after the detection/discovery of the UE unlike in the beam sweep case (broadened codebooks
lead to increasing uncertainty on the AoD that needs to be refined subsequently). For
an asymmetric downlink setup (as considered in this work), knowledge of AoD is more
important than that of AoA for the subsequent data delivery stage. Thus, it would be
of interest in understanding the latency tradeoffs between a broad beam scan (with
omni-directional scan as an extreme case) for the initial UE discovery stage followed by
beam/AoD refinement for the data delivery stage. While Fig.~8(b) of this work provides
a certain preliminary analysis, a more detailed study would be of interest.
Furthermore, architecturally speaking, an omni-directional scan would also render the
discovery of multiple MWBs impossible. Given that a mmW setup is expected to be primarily
of use in a small cell setting with a number of potential MWBs available for handoff,
lack of knowledge of AoD in an omni-directional scan could be potentially a disadvantage.

\section{Concluding Remarks}
We studied the efficacy of different beamforming approaches for initial UE discovery
in mmW MIMO systems in this work. The structure of the (near-)optimal beamformers
suggested RSV learning as a useful strategy for beamformer learning. However, RSV
learning is sensitive to small path length changes, a problem of serious importance
at mmW carrier frequencies. A further examination of the beamformer structure suggests
direction (AoA/AoD) learning as a viable strategy. We started by studying the utility
of classical approaches such as MUSIC in the context of direction learning. As with
RSV learning, MUSIC also requires a non-broadcast system design that could render it
unattractive from a system level standpoint. An alternate strategy motivated by the
limited feedback approach at cellular frequencies (albeit a directional codebook in 
the mmW context) of ${\sf SNR}$ estimation via the
use of a codebook of beamforming vectors at the MWB and UE is seen to result in a
broadcast solution that is conducive for initial UE discovery. While this approach has
a slightly poor performance relative to RSV learning and MUSIC, its simplicity
overweighs its sub-optimality. Table~\ref{table1} provides a brief summary of the
features of the different beamforming algorithms such as computational complexity,
${\sf PAR}$ of beamforming vector, system design issues and scaling with different
beamformer architectures.

\begin{table*}[htb!]
\caption{Qualitative Comparison Between Different Beamforming Approaches}
\label{table1}
\begin{center}
\begin{tabular}{|c||c|c|c|c|}
\hline
Issue of interest & RSV Learning & MUSIC/ESPRIT & Compressive Sensing & Beam Sweep \\
\hline
\hline
Computational& Iterative method & Computing eigenvectors &
Convex optimization & Received ${\sf SNR}$ \\
complexity & & & or similar & computation \\
\hline
${\sf PAR}$ of & Non-constant & CPO beam for & Non-constant & CPO \\
beamforming vector & amplitude & training & amplitude for & beam \\
& & & good dictionary & \\
\hline
Performance & Poor & Poor & Reasonable & Reasonable \\
robustness & & & & \\
\hline
System design & Unicast & Unicast & Broadcast & Broadcast \\
issues & Bi-directional & Bi-directional & Uni-directional & Uni-directional \\
\hline
Scaling with & Poor & Poor & Poor with a & Comparable \\
analog beamforming & & & general dictionary & \\
\hline
\end{tabular}
\end{center}
\end{table*}

\section*{Acknowledgments}
The authors would like to thank Prof.\ Sundeep Rangan (the Associate Editor), the
reviewers, Dr.\ Saurabh Tavildar, Dr.\ Omar El Ayach and Dr.\ Tianyang Bai for 
providing constructive comments that helped in improving the exposition of this paper.

\appendix
\subsection{Proof of Theorem~\ref{thm1}}
\label{appA}

The $N_t \times N_t$ matrix ${\sf H}^H {\sf H}$ can be expanded as
\begin{eqnarray}
\frac{L}{N_t N_r} \cdot {\sf H}^H {\sf H} & = &
\sum_{i,j} 
\alpha_i^{\star} \alpha_j \cdot \left( {\bf u}_i^H {\bf u}_j \right) \cdot
{\bf v}_i {\bf v}_j^H
\label{eq_HhermH}
= {\sf V} \hsppp {\sf A} \hsppp {\sf V}^H
\end{eqnarray}
where ${\sf V} = \left[ \alpha_1^{\star} \hsppp {\bf v}_1, \hsppp \cdots , \hsppp
\alpha_L^{\star} \hsppp {\bf v}_L \right]$ and ${\sf A}(i,j) = {\bf u}_i^H {\bf u}_j,
\hsppp i, j = 1, \cdots, L$.
Let ${\sf X}$ be an $L \times L$ eigenvector matrix of ${\sf A} \hsppp
{\sf V}^H \hsppp {\sf V}$ with the corresponding diagonal matrix of eigenvalues
denoted by ${\sf D}$. That is (the eigenvalue equation is given as),
\begin{eqnarray}
\left( {\sf A} \hsppp {\sf V}^H \hsppp {\sf V} \right) \cdot {\sf X} =
{\sf X} \cdot {\sf D}.
\label{eq_premultiply_eigs}
\end{eqnarray}
Pre-multiplying both sides of~(\ref{eq_premultiply_eigs}) by ${\sf V}$, we have
\begin{eqnarray}
{\sf V} \hsppp {\sf X} \cdot {\sf D} =
\left( {\sf V} \hsppp {\sf A} \hsppp {\sf V}^H \hsppp {\sf V} \right) \cdot {\sf X} =
\left( \frac{L}{N_t N_r} \cdot {\sf H}^H {\sf H}  \right) \cdot
{\sf V} \hsppp {\sf X}.
\label{eq_premultiply_eigs1}
\end{eqnarray}
Reading~(\ref{eq_premultiply_eigs1}) from right to left, we see that
${\sf V} \hsppp {\sf X}$ forms the eigenvector matrix for ${\sf H}^H {\sf H}$ with the
diagonal eigenvalue matrix being the same as ${\sf D}$. In other words, all the
eigenvectors of ${\sf H}^H {\sf H}$ (and hence ${\bf f}_{\sf opt}$) can be represented
as {\em linear combinations} of ${\bf v}_1, \cdots , {\bf v}_L$.
The only difference between the $L \leq N_t$ and $L > N_t$ cases is that the number
of distinct eigenvectors of ${\sf X}$ is less than or equal to $L$ and $N_t$ in the
two cases, respectively.

Given the structure of ${\bf f}_{\sf opt} = \sum_{j = 1}^L \beta_j {\bf v}_j$, we have
\begin{eqnarray}
{\sf H} \hsppp {\bf f}_{\sf opt} = \left ( \sum_{i = 1}^L \alpha_i {\bf u}_i {\bf v}_i^H \right)
\cdot \left( \sum_{j = 1}^L \beta_j {\bf v}_j \right)
= \sum_{i = 1}^L \alpha_i \cdot \left( \sum_j \beta_j {\bf v}_i^H {\bf v}_j \right)
{\bf u}_i
\label{eq43}
\end{eqnarray}
and thus ${\bf g}_{\sf opt}$ is a linear combination of $\{ {\bf u}_1, \cdots,
{\bf u}_L \}$.
\qed

\subsection{Proof of Theorem~\ref{thm2}}
\label{appB}
Since ${\sf SNR}_{\sf rx} = \rho_{\sf f} \cdot
\frac{ |{\bf g}^H {\sf H} {\bf f} |^2} { {\bf g}^H {\bf g} }$ is invariant to
$\| {\bf g} \|_2$, we let ${\bf g} = \| {\bf g} \|_2
\cdot {\widetilde{\bf g}}$ where $\| {\widetilde{\bf g} } \|_2 = 1$ (Constraint 1).
The $2$- and $\infty$-norm constraints on ${\bf g}$ are equivalent to $\| {\bf g} \|_2
\leq 1$ and $\| {\bf g} \|_{\infty} = \| {\bf g} \|_2 \cdot \| {\widetilde{\bf g}}
\|_{\infty} \leq \frac{1}{ \sqrt{N_r}}$ (Constraints 2 and 3). The received ${\sf SNR}$
maximization can then be recast as
\begin{eqnarray}
\widetilde{\bf g}_{\sf opt} = \arg \max_{ \widetilde{\bf g}}  \frac{ {\sf SNR}_{\sf rx} }
{\rho_{\sf f} } = \arg \max_{ \widetilde{\bf g}} \frac{ |\widetilde{\bf g}^H {\sf H} {\bf f}|^2 }
{ \widetilde{\bf g}^H \widetilde{\bf g} }
\label{eq44}
\end{eqnarray}
subject to Constraints 1 to 3. Ignoring Constraints 2 and 3 in the above optimization,
${\sf SNR}_{\sf rx}$ can be upper bounded as $\rho_{\sf f} \cdot  \| {\sf H}{\bf f}\|_2^2$
with equality if $\widetilde{\bf g} = \alpha \cdot {\sf H} {\bf f}$ for some $\alpha$.
We now consider a specific choice ${\bf g} = \frac{1}{ \sqrt{N_r} } \cdot \frac{
{\sf H} {\bf f} }{ \| {\sf H} {\bf f} \|_{\infty} }.$ We note that this choice is in
${\cal F}_{\infty}^{N_r}$ by satisfying all the three constraints since $\widetilde{\bf g}
= \frac{ {\sf H}{\bf f} }{ \| {\sf H} {\bf f} \|_2 }$, $\| {\bf g} \|_2 =
\frac{ \| {\sf H} {\bf f} \|_2 }{ \sqrt{N_r} \cdot \| {\sf H} {\bf f} \|_{\infty} }
\leq 1$ and $\| {\bf g} \|_2 \cdot \| {\widetilde{\bf g}} \|_{\infty} = \frac{1}{
\sqrt{N_r} }$. Further, the upper bound for ${\sf SNR}_{\sf rx}$ is also met and this
choice is optimal from ${\cal F}_{\infty}^{N_r}$.


Let ${\bf f}$ be the optimal beamformer from ${\cal F}_{\infty}^{N_t}$ that maximizes
${\bf f}^H {\sf H}^H{\sf H} {\bf f}$ and let the magnitude of at least one of its entries
not equal $\frac{1}{\sqrt{N_t}}$. Without loss in generality (by appropriate rotations with
permutation matrices), let one of these entries be the first entry. Let ${\bf f}$ and
${\sf H}^H {\sf H}$ be partitioned as
\begin{eqnarray}
{\bf f} = \left[
\begin{array}{c}
\underbrace{ {\bf f}(1) }_{1 \times 1} \\
\underbrace{ {\bf f}_{\sf r} }_{ N_t - 1 \times 1 }
\end{array}
\right], \hspp
{\sf H}^H {\sf H} = \left[
\begin{array}{cc}
\underbrace{ {\sf h} }_{1 \times 1} & \underbrace{ {\sf h}_{\sf r}^H }_{1 \times N_t - 1} \\
\underbrace{ {\sf h}_{\sf r} }_{ N_t - 1 \times 1} & \underbrace{ {\sf H}_{\sf r} }
_{ N_t - 1 \times N_t - 1}
\end{array}
\right].
\end{eqnarray}
With this partition, we have the following expansion:
\begin{eqnarray}
{\bf f}^H {\sf H}^H{\sf H} {\bf f} =
|{\bf f}(1)|^2 \hsppp {\sf h} +
2 \hsppp {\sf Re} \left( {\bf f}(1) \cdot {\bf f}_{\sf r}^H {\sf h}_{\sf r} \right) +
{\bf f}_{\sf r}^H {\sf H}_{\sf r} {\bf f}_{\sf r}. 
\label{eq_exp1}
\end{eqnarray}
Since ${\bf f}$ is optimal, whatever be the choice of ${\bf f}_{\sf r}$, we should have
$\angle{ {\bf f}(1) } = \angle{ {\bf h}_{\sf r}^H {\bf f}_{\sf r} }$. Otherwise, we can find
a better choice of ${\bf f}(1)$ for the same ${\bf f}_{\sf r}$. With this optimal choice
for $\angle{ {\bf f}(1) }$,~(\ref{eq_exp1}) reduces to
\begin{eqnarray}
{\bf f}^H {\sf H}^H{\sf H} {\bf f} =
|{\bf f}(1)|^2 \hsppp {\sf h} + 2 |{\bf f}(1)| \cdot | {\bf f}_{\sf r}^H {\bf h}_{\sf r} | +
{\bf f}_{\sf r}^H {\sf H}_{\sf r} {\bf f}_{\sf r}.
\label{eq_exp2}
\end{eqnarray}
Clearly,~(\ref{eq_exp2}) is increasing in $|{\bf f}(1)|$ and thus $|{\bf f}(1)| = \frac{1}{ \sqrt{N_t} }$.
Thus, the assumption that the magnitude of at least one of the entries of ${\bf f}$ is not
$\frac{1}{ \sqrt{N_t} }$ is untenable and we end up with a contradiction to the statement. This
implies that equal gain transmit beamforming is optimal.
\qed

\subsection{Proof of Prop.~\ref{prop1}}
\label{appC}
Let ${\bf f}(i)  = \frac{1}{ \sqrt{N_t} } \cdot e^{j \theta_i}$ for some $\theta_i$.
With the matched filter structure for ${\bf g}$ and with ${\sf H}$ as in the statement
of the proposition, ${\bf f}^H{\sf H}^H {\sf H}{\bf f}$ is given as
\begin{eqnarray}
\frac{ {\sf SNR}_{\sf rx}} { \rho_{\sf f}} = {\bf f}^H {\sf H}^H {\sf H} {\bf f} & = &
\sum_{i, k  = 1}^{N_t} {\sf h}_i^H {\sf h}_k \cdot e^{j (\theta_k - \theta_i)}
\label{eq45}
\\
& = & \sum_{i = 1}^{N_t} {\sf h}_i^H {\sf h}_i + 2 \sum_{ k < i }
{\sf Re} \left( {\sf h}_i^H {\sf h}_k \cdot e^{j (\theta_k - \theta_i)} \right).
\label{eq46}
\end{eqnarray}
While the optimal beamforming vector requires a simultaneous optimization over
$\{ \theta_i \}$, we can rewrite the objective function as
\begin{eqnarray}
{\bf f}^H {\sf H}^H {\sf H} {\bf f} =
\sum_{i = 1}^{N_t} {\sf h}_i^H {\sf h}_i +
2 \sum_{i = 1}^{N_t} {\sf Re} \left( e^{-j\theta_i} \cdot
\left( \sum_{k = 1}^{i-1}e^{j \theta_k} \cdot {\sf h}_i^H {\sf h}_k
\right) \right).
\label{eq_34}
\end{eqnarray}
With a recursive structure as in~(\ref{eq18}), each term in~(\ref{eq_34}) is
maximized, even though the impact of this structure on the sum of these terms
is unclear.
\qed

\subsection{Proof of Theorem~\ref{thm3}}
\label{appD}
With $\{ \Omega_j \}$ as sampling frequencies, we have
\begin{eqnarray}
\min_{ j = 1, \cdots, J} | {\bf F}(\Omega_j)|^2 & \leq &
\frac{1}{J} \cdot \sum_{j = 1}^J  | {\bf F}(\Omega_j)|^2 =
\frac{1}{J} \cdot {\bf f}^H \left( \sum_{j = 1}^J {\bf a}(\Omega_j)
{\bf a}(\Omega_j)^H \right) {\bf f}
\label{eq48}
\\
& \leq & \frac{ \lambda_{\sf max} \left(
\sum_{j = 1}^J {\bf a}(\Omega_j) {\bf a}(\Omega_j)^Hi \right)}{J}
\label{eq_36}
\end{eqnarray}
where the last inequality follows from the Rayleigh coefficient of a Hermitian
matrix and its largest eigenvalue. A straightforward argument shows that the $J = 2$
setting minimizes the above upper bound when $\Omega_0 = \frac{ {\sf A} \cdot 2 \pi}
{N_t}$ where $0 \leq {\sf A} \leq 1$. Further, in this case, the choice $\Omega_2 =
\frac{ \pi}{2} + \frac{ {\sf A} \cdot \pi}{N_t}$ and $\Omega_1 = \frac{\pi}{2} -
\frac{ {\sf A} \cdot \pi}{N_t}$ minimizes the upper bound resulting in
\begin{eqnarray}
\lambda_{\sf max} \left(
\sum_{j = 1}^J {\bf a}(\Omega_j) {\bf a}(\Omega_j)^H \right) =
N_t + \Bigg| \frac{ \sin( N_t(\Omega_2 - \Omega_1)/2) }
{ \sin( (\Omega_2 - \Omega_1)/2) } \Bigg|.
\label{eq50}
\end{eqnarray}
Using this fact, the worst-case beamforming gain is seen to be
\begin{eqnarray}
{\sf BF \hspp Gain} \leq 10 \log_{10} \left( \min \left( N_t,
\hspp \frac{N_t}{2} + \frac{1}{2} \cdot
\left| \frac{ \sin \left( \frac{ \overline{\Omega} \cdot N_t }
{2 \hsppp {\sf No. \hsppp beams} } \right) }
{ \sin \left( \frac{ \overline{\Omega}  }
{2 \hsppp {\sf No. \hsppp beams} } \right) }
\right| \right) \right)
\label{eq51}
\end{eqnarray}
provided ${\sf No. \hsppp beams} \geq \frac{ \overline{\Omega} \cdot N_t} {2 \pi}$.
In general, it is difficult to obtain closed-form expressions for the eigenvalues of
$\sum_{j = 1}^J {\bf a}(\Omega_j) {\bf a}(\Omega_j)^H$ with $J = 3,4$ and this task
is impossible if $J > 4$. Thus, in these settings (${\sf No. \hsppp beams} <
\frac{ \overline{\Omega} \cdot N_t}{2 \pi}$), the best-case (smallest) upper bound is
obtained by minimizing the quantity in~(\ref{eq_36}) over different choices of $J$
and $\{ \Omega_j \}$ leading to~(\ref{eq_23}), where $\Omega_j \in
\left[ \Omega_{\sf c} - \frac{\Omega_0}{2}, \hspp \Omega_{\sf c} + \frac{ \Omega_0}{2}
\right]$ with $\Omega_0 = \frac{ \overline{\Omega}}{ {\sf No. \hsppp beams}} >
\frac{2 \pi}{N_t}$.
\qed

{\vspace{-0.05in}}
\bibliographystyle{IEEEbib}
\bibliography{newrefsx}

\end{document}